\documentclass[a4paper,onecolumn,11pt,unpublished]{quantumarticle}
\pdfoutput=1
\usepackage[utf8]{inputenc}
\usepackage[english]{babel}
\usepackage[T1]{fontenc}
\usepackage{hyperref}

\usepackage{tikz}
\usepackage{lipsum}

\usepackage{amsmath,amssymb,amsfonts,amsthm}
\usepackage{mathrsfs}%
\usepackage{xcolor}%
\usepackage{textcomp}%
\usepackage{booktabs}%
\usepackage{listings}%

\usepackage{adjustbox}
\usepackage{graphicx}%
\usepackage{multirow}%
\usepackage{makecell}
\usepackage{array}
\usepackage{multirow}

\usepackage{subfigure}
\usepackage{physics}
\usepackage{nomencl}
\makenomenclature
\usepackage{adjustbox}

\usepackage{amsmath,amsfonts,bm}

\def\eqref#1{Eq.~(\ref{#1})}

\def\1{\bm{1}}

\DeclareMathAlphabet{\mathsfit}{\encodingdefault}{\sfdefault}{m}{sl}
\SetMathAlphabet{\mathsfit}{bold}{\encodingdefault}{\sfdefault}{bx}{n}

\newtheorem{definition}{Definition}
\newtheorem{theorem}{Theorem}
\newtheorem{lemma}{Lemma}

\newtheorem{assumption}{Assumption}

\usepackage{comment}
\usepackage{soul}
\usepackage{url}

\begin{document}

\title{Embedding-Aware Noise Modeling of Quantum Annealing}

\author{Seon-Geun Jeong}
\orcid{0000-0002-5926-6642}
\email{wjdtjsrms11@pusan.ac.kr}

\author{Mai Dinh Cong}
\email{cong.md@pusan.ac.kr}
\orcid{0009-0003-3730-8869}
\affiliation{Department of Information Convergence Engineering, Pusan National University, 46241 Busan, Republic of Korea}

\author{Dae-Il Noh}
\orcid{0000-0002-6586-5780}
\email{daeil.noh@kqchub.com}
\affiliation{Quantum AI Lab, Korea Quantum Computing Co., Ltd, 48058 Busan, Republic of Korea}

\author{Quoc-Viet Pham}
\orcid{0000-0002-9485-9216}
\email{daeil.noh@kqchub.com}
\affiliation{School of Computer Science and Statistics, Trinity College Dublin, Dublin 2, D02 PN40 Dublin, Ireland}

\author{Won-Joo Hwang}
\orcid{0000-0001-8398-564X}
\email{wjhwang@pusan.ac.kr}
\affiliation{School of Computer Science and Engineering, Pusan National University, 46241 Busan, Republic of Korea}

\maketitle

\begin{abstract}
Quantum annealing provides a practical realization of adiabatic quantum computation and has emerged as a promising approach for solving large-scale combinatorial optimization problems. However, current devices remain constrained by sparse hardware connectivity, which requires embedding logical variables into chains of physical qubits. This embedding overhead limits scalability and reduces reliability as longer chains are more prone to noise-induced errors. In this work, building on the known structural result that the average chain length in clique embeddings grows linearly with the problem size, we develop a mathematical framework that connects embedding-induced overhead with hardware noise in D-Wave’s Zephyr topology. Our analysis derives closed-form expressions for chain break probability and chain break fraction under a Gaussian control error model, establishing how noise scales with embedding size and how chain strength should be adjusted with chain length to maintain reliability. Experimental results from the Zephyr topology-based quantum processing unit confirm the accuracy of these predictions, demonstrating both the validity of the theoretical noise model and the practical relevance of the derived scaling rule. Beyond validating a theoretical model against hardware data, our findings establish a general embedding-aware noise framework that explains the trade-off between chain stability and logical coupler fidelity. Our framework advances the understanding of noise amplification in current devices and provides quantitative guidance for embedding-aware parameter tuning strategies.
\end{abstract}

\section{Introduction}
Quantum computing is widely recognized as a promising technology across diverse application domains~\cite{jeong2025hybridquantumneuralnetworks, moon2025qsegrnn, 10981540, 10907925, 11177505, 11091511}.
In particular, quantum annealing (QA) is a practical realization of adiabatic quantum computation (AQC)~\cite{albash2018adiabatic}, leveraging quantum tunneling and thermal effects to explore complex energy landscapes~\cite{morita2008mathematical}. While AQC provides a theoretical guarantee of convergence under sufficiently slow evolution, QA enables implementation in near-term devices and has been commercialized by D-Wave Systems~\cite{dwave, johnson2011quantum, king2022coherent}. Through successive hardware generations such as Chimera, Pegasus, and most recently Zephyr~\cite{boothby2020next, Boothby2021Zephyr, dattani2019pegasus, zbinden2020embedding}, D-Wave has steadily increased the number of qubits and improved connectivity, making QA a promising approach for solving combinatorial optimization problems on real hardware~\cite{10907925}.

Despite this progress, a persistent challenge lies in the embedding process. Sparse hardware connectivity requires logical problem graphs to be mapped onto physical qubits through chains, causing an overhead that reduces both scalability and solution reliability~\cite{pelofske20244, fang2020minimizing}. Clique embeddings, which encode fully connected graphs, have become a standard benchmark for probing hardware capacity because they capture the worst-case embedding requirement~\cite{boothby2016fast}. Prior empirical studies have measured clique sizes, approximation ratios, chain break fractions (CBFs), and time-to-solution across hardware generations~\cite{pelofske2025comparing, konz2021embedding}. Other works have characterized chain breaking phenomena, identifying empirical sweet spots for intra-chain couplings and correlations between chain break locations and embedding structures~\cite {grant2022benchmarking}. While these studies have provided valuable insights, they are fundamentally empirical and do not yield predictive models of how embedding overhead interacts with hardware noise. Recent analyses of embedding overhead scaling further emphasize that minor embeddings impose intrinsic limits on device capacity~\cite{konz2021embedding}.

While previous works have shown that chain length grows linearly with the grid parameter $m$~\cite{boothby2020next, Boothby2021Zephyr}, the implications for noise amplification and chain stability have not been rigorously quantified. Prior studies on minor embedding and clique embeddings~\cite{boothby2016fast, venturelli2015quantum, choi2008minor, choi2011minor} primarily focused on structural aspects, such as chain length and the effect of intra-chain coupling strength. However, chain breaking is ultimately a noise-driven phenomenon, and thus requires a statistical error model to move beyond purely structural analyses. The integrated control error (ICE) model, commonly introduced in D-Wave documentation, provides a convenient statistical description of local field and coupler errors~\cite{dwave}, while decoherence analyses in QA suggest Gaussian perturbations as a reasonable approximation~\cite{albash2015decoherence}. However, ICE has not been rigorously validated in the context of embedding overhead, nor has it been extended into a predictive framework. In particular, the relationship between chain length growth, control error propagation, and observable chain break statistics remains unexplored. This gap is critical because without a predictive model, it would be very challenging to anticipate how embedding overhead will constrain scalability or to design principled noise-mitigation strategies as the problem size increases~\cite{shingu2024quantum, raymond2025quantum}.

Our framework addresses this limitation by moving beyond empirical benchmarking to a rigorous analytic treatment, linking structural embedding properties to quantifiable chain break statistics. Specifically, we derive and experimentally validate closed-form scaling laws for chain break probability (CBP) and CBF. Predicting CBF is especially valuable, as it directly connects embedding overhead to the practical reliability of QA, enabling scalability assessments, parameter tuning (e.g., intra-chain coupling strength, annealing time), and embedding-aware QA simulation~\cite{grant2022benchmarking, le2023benchmarking, bando2021simulated}. The alignment between theoretical predictions and experimental results not only explains observed chain break statistics but also provides a foundation for embedding-aware design principles and forward-looking scalability analysis of future quantum annealers.

The contributions of this paper are as follows: 
\begin{itemize}
    \item We develop an embedding-aware noise model that extends the ICE framework by explicitly linking chain length growth in clique embeddings with chain break statistics.
    \item We derive closed-form scaling laws for variance accumulation, CBP, and CBF, thereby establishing how embedding overhead directly translates into reliability degradation.
    \item We validate the proposed embedding-aware noise model through systematic experiments on the D-Wave Advantage2 processor (Zephyr topology) with clique embedding. The results show close concordance between theoretical predictions and hardware-observed CBF and energy statistics.
    \item We establish a theoretical scaling rule that prescribes how the chain strength should grow with embedding size to maintain reliability, and we empirically confirm this rule on hardware. At the same time, our experiments reveal a steeper but still sublinear growth than the ideal prediction, which we attribute to correlated hardware noise. This contrast between the idealized theoretical law and the practical scaling behavior highlights the importance of embedding-aware noise modeling.
    \item We position the resulting framework as a general embedding-aware noise model, offering predictive tools for assessing quantum annealer scalability and informing embedding-aware, hardware-conscious noise modeling strategies.
\end{itemize}

The remainder of this paper is organized as follows. Section~\ref{sec:background} reviews the background of QA and D-wave quantum processing unit (QPU) topology. Section~\ref{sec:scaling} develops the analytical scaling law for chain length growth in clique embeddings and derives closed-form expressions for CBP and CBF. Section~\ref{sec:experiments} presents experimental validation using D-Wave Advantage2 hardware. Section~\ref{sec:discussion} interprets the results by analyzing sensitivity to annealing schedules and chain strength, and further relates the proposed model to other physical noise channels beyond ICE. Finally, Section~\ref{sec:conclusion} concludes with a discussion of implications for scalability and embedding-aware noise mitigation in future quantum annealers.

\section{Background}\label{sec:background}

\subsection{Quantum Annealing}

The QA is a heuristic realization of AQC, where the objective is to minimize an Ising Hamiltonian that encodes a combinatorial optimization problem~\cite{kadowaki1998quantum,das2008colloquium}. The time-dependent Hamiltonian is defined as
\begin{equation}\label{eq:qaorigin}
    H(t) = A(t) \sum_i \sigma_i^x + B(t) \left( \sum_i h_i \sigma_i^z + \sum_{i<j} J_{ij} \sigma_i^z \sigma_j^z \right),
\end{equation}
where $\sigma_i^x$ and $\sigma_i^z$ are Pauli operators, $h_i$ and $J_{ij}$ represent programmable the local field on qubit $i$ and the coupler between qubits $i$ and $j$, respectively, and $A(t), B(t)$ are annealing schedules satisfying $A(0) \gg B(0)$ and $A(T) \ll B(T)$. In the adiabatic limit, the system remains in its instantaneous ground state and ideally reaches the optimal solution of the encoded problem. In practice, finite annealing times, thermal fluctuations, and hardware noise yield a heuristic sampling process rather than guaranteed convergence.

\subsection{D-Wave QPU Topologies}

D-Wave quantum annealers implement QA using superconducting flux qubits arranged in sparse but structured topologies. The hardware connectivity has evolved from the Chimera graph (degree~$6$) to the Pegasus topology (degree~$15$) and most recently to the Zephyr topology (degree~$20$)~\cite{boothby2020next,Boothby2021Zephyr}. These topologies define how logical problem graphs can be embedded into the physical hardware. Figure~\ref{fig:zephyr_graphs} illustrates the Zephyr connectivity in both the real Advantage2 system and a toy instance.

\begin{figure}[!t]
\centering
\subfigure[Zephyr hardware graph (\textit{Advantage2\_system1.6}).]{
    \includegraphics[width=0.45\textwidth]{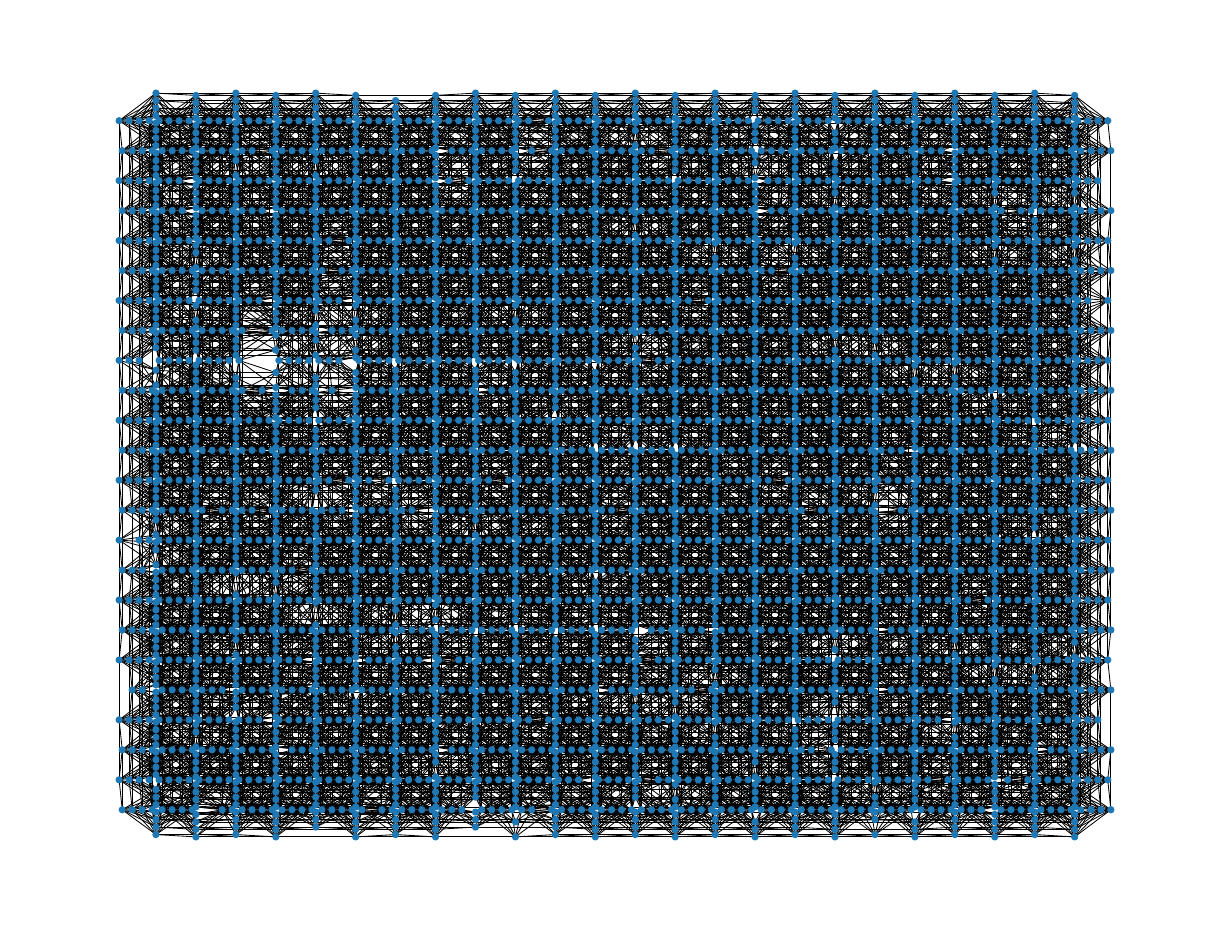}
    \label{fig:zephyr_real}
}
\subfigure[Zephyr toy graph ($m=1,\, t=4$).]{
    \includegraphics[width=0.45\textwidth]{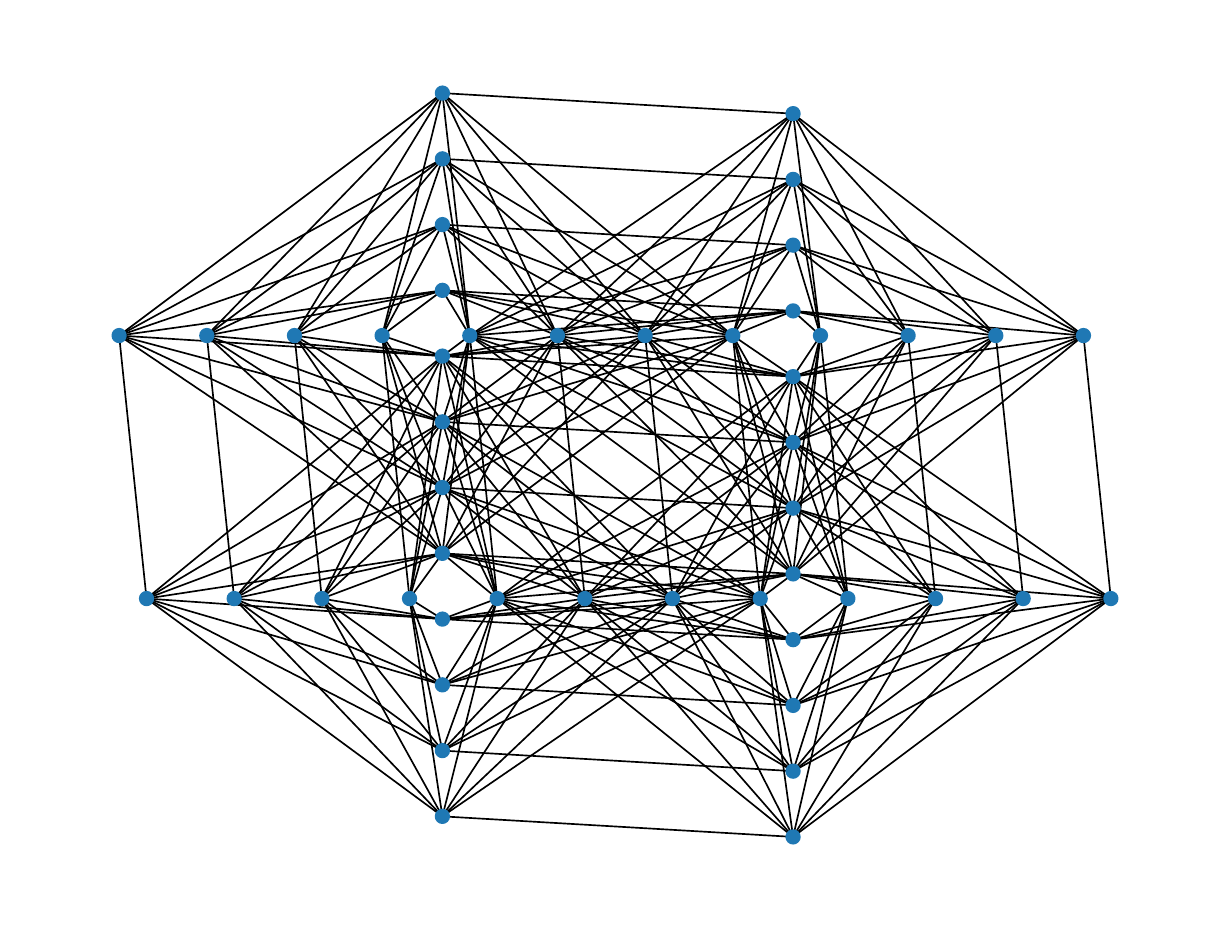}
    \label{fig:zephyr_toy}
}
\caption{Illustration of Zephyr topology. (a) shows the actual hardware connectivity of the Advantage2\_system1.6 QPU, while (b) shows a toy Zephyr graph $Z_{1,4}$ for clarity of unit-cell structure.} 
\label{fig:zephyr_graphs}
\end{figure}

Formally, a Zephyr graph $Z_{m,t}$ with grid parameter $m$ and tile size $t$ contains
\begin{equation}
    |V(Z_{m,t})| = 8tm^2 + 4tm,
\end{equation}
qubits, each of maximum degree $4t+4$. For example, $Z_{15,4}$ contains $7440$ qubits. Each qubit is identified by coordinates $(u,w,k,j,z)$, where $u \in \{0,1\}$ denotes orientation (vertical or horizontal), $w$ and $z$ are block offsets, $k$ is a within-tile index, and $j$ indicates a shift~\cite{Boothby2021Zephyr}. Zephyr supports three classes of couplers: internal, connecting orthogonal qubits; external, connecting qubits aligned along rows or columns; and odd, connecting parallel qubits in adjacent rows or columns. This enhanced connectivity enables richer native subgraphs, including $K_4$ (complete graphs on four nodes) and $K_{8,8}$ (complete bipartite graphs with partitions of size eight) as substructures of the hardware graph.

Despite these improvements, the Zephyr topology remains far from fully connected, making minor embedding necessary for dense logical problems~\cite{pelofske20244}. In this process, each logical variable is represented by a chain of ferromagnetically coupled physical qubits. Clique embeddings, which realize complete graphs $K_n$ within the hardware, serve as a standard benchmark for embedding capacity~\cite{boothby2016fast,grant2022benchmarking}. In Zephyr, cliques up to $K_{16m-8}$ can be embedded, and the required chain length grows proportionally to $m$ as shown in the construction of Boothby et al.~\cite{Boothby2021Zephyr}. While increased connectivity shortens average chain length compared to Pegasus and Chimera, the linear scaling introduces a fundamental overhead: as the number of logical variables grows, chains become longer and thus more vulnerable to errors.

A chain break occurs when physical qubits within a chain disagree, preventing the chain from behaving as a single logical qubit. The probability of such breaks depends strongly on the chain strength $k$ (the intra-chain ferromagnetic penalty used to enforce chain alignment): if $k$ is too weak, chains frequently break; if $k$ is too strong, logical couplers $J_{ij}$ are suppressed, reducing problem fidelity~\cite{grant2022benchmarking,venturelli2015quantum}. Previous empirical studies have identified sweet spots for $k$ that balance chain stability and logical accuracy, but these analyses remain heuristic.

In summary, prior work has established the structural embedding capacity of Zephyr, demonstrating that chain length scales linearly with the problem size. However, the implications of such growth for noise accumulation and chain stability have not been mathematically analyzed. This motivates our study, which links embedding-induced chain length scaling with error propagation and develops predictive scaling laws for chain reliability under realistic noise models.

\section{Analytical Scaling Law}\label{sec:scaling}
In this section, we derive a mathematical framework that connects embedding-induced overhead to noise amplification. 
We begin by introducing the setup and assumptions, including the Gaussian error model and a chain stability criterion. 
Under the assumption that local field and coupler errors are independent Gaussian perturbations, the variance of accumulated errors grows linearly with the chain length. 
Building on this variance law, we approximate the CBP using Gaussian tail probabilities and relate it to the CBF. 
Finally, we express these results in terms of the Zephyr grid parameter $m$, highlighting the scalability limitations and their implications for selecting the chain strength.

\subsection{Embedding and Assumption}
We first describe the embedding setup and introduce the assumptions that underlie our analysis. 
These include the ICE-based Gaussian noise model for local control errors, the approximate uniformity of chain lengths in clique embeddings, 
and the criterion for declaring a chain break. 
Together, these assumptions provide the foundation for deriving variance growth and break probabilities.

\begin{definition}[Minor Embedding and Chains]
Let $G_L = (V_L, E_L)$ denote the logical graph and $G_P = (V_P, E_P)$ denote the hardware graph.  
A \emph{minor embedding} is a mapping $\phi : V_L \to \{T_i \subseteq V_P\}$ such that each logical vertex $i \in V_L$ is represented by a connected, nonempty subset of physical qubits $T_i$, referred to as a \emph{chain}.  
Each logical edge $(i,j) \in E_L$ is realized by at least one physical edge between a qubit in $T_i$ and a qubit in $T_j$, and the chains are pairwise disjoint, i.e., $T_i \cap T_j = \emptyset$ for all $i \neq j$.  
The chain length associated with vertex $i$ is $\ell_i := |T_i|$, and $\bar{\ell}$ denotes the average chain length across all logical vertices.
\end{definition}

\begin{definition}[Embedded Hamiltonian~\cite{grant2022benchmarking}]\label{def:emHamil}
Given logical fields $\{h_i\}$ and couplers $\{J_{ij}\}$, the embedded Ising Hamiltonian programmed onto the hardware is
\begin{equation}
    H^* \;=\; -\sum_{p\in V_P} h^*_{p}\,\sigma^z_p 
              \;-\; \sum_{(p,q)\in E_P} J^*_{pq}\,\sigma^z_p\sigma^z_q,
\end{equation}
with physical parameters
\begin{equation}
    h^*_{p}=\frac{h_i}{|T_i|}\;\;(p\in T_i), \qquad 
    J^*_{pq}=\begin{cases}
        \frac{J_{ij}}{|\mathrm{edges}(T_i,T_j)|}, & p\in T_i,\,q\in T_j,\, i\neq j,\\[4pt]
        k, & p,q\in T_i \;\;(\text{intra-chain penalty}),\\
        0, & \text{otherwise}.
    \end{cases}
\end{equation}
Here, $k>0$ is the chain strength, i.e., the intra-chain ferromagnetic penalty enforcing chain consistency. 
\end{definition}

\begin{assumption}[ICE-based Noise Model] \label{Ass:NoiseModel}
While $H^*$ denotes the programmed Hamiltonian, 
the actual Hamiltonian realized on the QPU is subject to ICE~\cite{dwave,yarkoni2022quantum}. 
That is, instead of $H^*$, the device effectively implements
\begin{equation}
    H^{\delta} = - \sum_{p\in V_P} (h^*_{p} + \delta h_p)\,\sigma_p^z 
                 - \sum_{(p,q)\in E_P} (J^*_{pq} + \delta J_{pq})\,\sigma_p^z \sigma_q^z,
    \label{eq:ice-hamiltonian}
\end{equation}
where $\delta h_p$ and $\delta J_{pq}$ denote deviations between programmed and realized parameters. 
For analytical tractability, we restrict attention to the misspecification component of ICE and approximate the perturbations as independent zero-mean Gaussians:
\begin{equation}
    \delta h_p \sim \mathcal{N}(0,\sigma_h^2), 
    \qquad 
    \delta J_{pq}\sim \mathcal{N}(0,\sigma_c^2).
\end{equation}
This simplified Gaussian model has been widely used in prior analyses of decoherence and control noise~\cite{albash2015decoherence}.
\end{assumption}

\begin{assumption}[Uniform Chain Length for Cliques]\label{ass:uniform}
For clique embeddings on Zephyr at grid parameter $m$, chain lengths are approximately uniform:
\begin{equation}
    \ell_i \;\approx\; \alpha m + \beta,
\end{equation}
where $\alpha$ and $\beta$ are topology-dependent constants.
\end{assumption}
This assumption is motivated by the construction of clique embeddings, which produce nearly uniform chain lengths up to small boundary variations~\cite{Boothby2021Zephyr,boothby2016fast}.

\begin{assumption}[Failure Criterion]\label{ass:failure}
A chain $T_i$ breaks if the accumulated effective error $\Delta(\ell_i)$ exceeds a stabilizing margin $k_{\mathrm{eff}}=\eta k$, with proportionality constant $\eta\in(0,1]$:
\begin{equation}
    \text{break if } |\Delta(\ell_i)| > k_{\mathrm{eff}}.
\end{equation}
\end{assumption}
This assumption captures the idea that chain stability is maintained 
as long as the chain strength $k$ 
dominates error perturbations, consistent with prior empirical observations of sweet spots for $k$~\cite{grant2022benchmarking,venturelli2015quantum}.

\subsection{Variance Scaling}
Having defined the noise model and break criterion, we next analyze how control errors accumulate along a chain. We prove that the variance of the total chain-level error grows linearly with chain length, formalizing the intuition that longer chains collect proportionally more noise.

Each logical variable is encoded as a chain of ferromagnetically coupled qubits. 
In practice, every qubit bias $h_p$ and intra-chain coupler $J_{pq}$ is affected by small control errors $\delta h_p$ and $\delta J_{pq}$. 
The quantity $\Delta(\ell_i)$ aggregates all such errors acting on a chain, and therefore represents the net perturbation competing against the stabilizing chain strength. 
If $\Delta(\ell_i)$ becomes comparable to or larger than $k$, qubits in the chain may disagree, leading to a chain break. 
Therefore, $\Delta(\ell_i)$ can be interpreted as the effective noise load applied to a logical variable due to hardware imperfections.
\begin{definition}[Chain-Level Control Error]
For a chain $T_i$ of length $\ell_i$, define the accumulated error as
\begin{equation}
    \Delta(\ell_i) \;=\; \sum_{p\in T_i}\delta h_p \;+\; \sum_{(p,q)\in E_P\cap(T_i\times T_i)} \delta J_{pq}.
\end{equation}
\end{definition}

\begin{theorem}[Variance Scaling]\label{the:var}
Under independent zero-mean Gaussian in Assumption~\ref{Ass:NoiseModel}, the variance of $\Delta(\ell_i)$ grows linearly with $\ell_i$:
\begin{equation}\label{the:var_scailing}
    \mathrm{Var}[\Delta(\ell_i)] \;\approx\; \ell_i\sigma_h^2 + (\ell_i-1)\sigma_c^2.
\end{equation}
\end{theorem}

\begin{proof}
$\Delta(\ell_i)$ is a sum of $\ell_i$ independent field errors and $\ell_i-1$ independent intra-chain coupler errors.  Each contributes variance $\sigma_h^2$ and $\sigma_c^2$, respectively. Under the assumption of independent Gaussian variables in Assumption \ref{Ass:NoiseModel}, we obtain the total variance as 
\begin{equation*}
    \mathrm{Var}[\Delta(\ell_i)] = \ell_i\sigma_h^2 + (\ell_i-1)\sigma_c^2.
\end{equation*}
\end{proof}
This result formalizes the idea that every additional qubit or coupler in a chain contributes an independent source of noise. 
As the chain grows longer, these contributions add in variance, leading to a linear increase in the total uncertainty. 
Thus, embedding overhead directly translates into noise overhead, which underpins the scaling of CBP derived in the next subsection.

\subsection{Chain Break Probability and Fraction}
The CBP characterizes the probability that a single embedded chain $T_i$ fails due to the accumulated control error exceeding the stabilizing chain strength $k$.
Similarly, the CBF measures the fraction of all embedded chains, each representing a logical variable, that break within a given embedding instance.
Therefore, CBP is a chain-level quantity determined by chain length and noise variance, while CBF is a global performance metric directly observable on hardware by counting broken chains across samples. 
\begin{definition}[Chain Break Probability]\label{def:cbp}
For a chain $T_i$ of length $\ell_i$, the CBP is
\begin{equation}
    \mathrm{CBP}(\ell_i) \;:=\; \Pr\big(|\Delta(\ell_i)| > k_{\mathrm{eff}}\big),
\end{equation}
where $\Delta(\ell_i)$ denotes the accumulated control error on $T_i$.
\end{definition}
\begin{definition}[Chain Break Fraction]\label{def:cbf}
For an embedding consisting of chains $\{T_1,\dots,T_{|V_L|}\}$, the CBF is
\begin{equation}
    \mathrm{CBF} := \frac{1}{|V_L|}\sum_{i=1}^{|V_L|}\mathbf{1}\{|\Delta(\ell_i)|>k_{\mathrm{eff}}\},
\end{equation}
where $\mathbf{1}\{\cdot\}$ denotes the indicator function, which returns $1$ if the condition inside is true and $0$ otherwise.
\end{definition}



By Definition~\ref{def:cbp}, CBP is the probability that the chain-level control error exceeds the stabilizing margin.
Approximating the accumulated error as Gaussian with variance $\mathrm{Var}[\Delta(\ell_i)]$ turns CBP into a Gaussian tail probability.

\begin{theorem}[Gaussian Tail Approximation of CBP]\label{th:gaussian-tail}
Under the Gaussian noise model,
\begin{equation}
    \mathrm{CBP}(\ell_i) \;\approx\; \operatorname{erfc}\!\left(\frac{k_{\mathrm{eff}}}{\sqrt{2\,\mathrm{Var}[\Delta(\ell_i)]}}\right),
\end{equation}
with $\mathrm{Var}[\Delta(\ell_i)] = \ell_i\sigma_h^2+(\ell_i-1)\sigma_c^2$.
\end{theorem}

\begin{proof}\label{pr:1}
From Definition~\ref{def:cbp}, the CBP is the probability that the accumulated error exceeds the stabilizing margin:
\begin{equation}
    \mathrm{CBP}(\ell_i) = \Pr\!\big(|\Delta(\ell_i)| > k_{\mathrm{eff}}\big).
\end{equation}
By the Gaussian noise assumption, $\Delta(\ell_i)$ is well approximated as a normal random variable with zero mean and variance $\mathrm{Var}[\Delta(\ell_i)]$:
\begin{equation}\label{eq:ref1}
    \Delta(\ell_i) \sim \mathcal{N}\!\left(0,\,\mathrm{Var}[\Delta(\ell_i)]\right).
\end{equation}
Thus, evaluating CBP reduces to computing the tail probability of a centered Gaussian. 
For a general zero-mean Gaussian random variable $X \sim \mathcal{N}(0,\sigma^2)$, one has
\begin{equation}
    \Pr(|X| > a) = 2\Pr(X > a) = \frac{2}{\sqrt{2\pi\sigma^2}}\int_a^{\infty} \exp\!\left(-\frac{t^2}{2\sigma^2}\right)\,dt.
\end{equation}
The integral on the right-hand side is exactly the complementary error function, defined as
\begin{equation}\label{eq:complem}
    \operatorname{erfc}(x) := \frac{2}{\sqrt{\pi}}\int_x^{\infty} e^{-u^2}\,du.
\end{equation}
By a simple change of variables $u = t/(\sqrt{2}\sigma)$, the Gaussian tail probability can be rewritten as
\begin{equation}
    \Pr(|X| > a) = \operatorname{erfc}\!\left(\frac{a}{\sqrt{2}\sigma}\right).
\end{equation}
Finally, substituting $a = k_{\mathrm{eff}}$ and $\sigma^2 = \mathrm{Var}[\Delta(\ell_i)]$ yields
\begin{equation*}
    \mathrm{CBP}(\ell_i) \;\approx\; \operatorname{erfc}\!\left(\frac{k_{\mathrm{eff}}}{\sqrt{2\,\mathrm{Var}[\Delta(\ell_i)]}}\right).
\end{equation*}
This completes the proof.
\end{proof}

This theorem formalizes the intuition that longer chains are more likely to break: as $\ell_i$ grows, the variance $\mathrm{Var}[\Delta(\ell_i)]$ increases linearly, reducing the effective signal-to-noise ratio $k_{\mathrm{eff}}/\sqrt{\mathrm{Var}[\Delta(\ell_i)]}$ and thereby increasing the Gaussian tail probability. 
The erfc function captures this trade-off quantitatively. 
It is noteworthy that this closed-form expression enables direct comparison with experimental data and forms the basis for deriving scaling laws for the chain strength $k$ to maintain constant reliability.

CBF measures the proportion of broken chains in a given embedding instance. 
It provides an embedding-level metric that complements CBP, which is defined at the level of a single chain.

\begin{lemma}[Approximation of CBF by CBP]
By Definition~\ref{def:cbf}, if the chain lengths are approximately uniform, i.e., $\ell_i \approx \bar\ell$ for all $i$, then the expected CBF satisfies
\begin{equation}
    \mathbb{E}[\mathrm{CBF}] \;\approx\; \mathrm{CBP}(\bar\ell).
\end{equation}
\end{lemma}

\begin{proof}
Since each term $\mathbf{1}\{|\Delta(\ell_i)|>k_{\mathrm{eff}}\}$ is a Bernoulli random variable with success probability $\mathrm{CBP}(\ell_i)$, taking expectations yields
\begin{equation}
    \mathbb{E}[\mathrm{CBF}] = \frac{1}{|V_L|}\sum_{i=1}^{|V_L|}\mathrm{CBP}(\ell_i).
\end{equation}
If $\ell_i\approx \bar\ell$ for all $i$, then $\mathrm{CBP}(\ell_i)\approx \mathrm{CBP}(\bar\ell)$, giving the result.
\end{proof}

\subsection{Scaling with Embedding Size}
Finally, we express CBP as a function of the Zephyr grid parameter $m$ using $\ell_i=\alpha m+\beta$ from Assumption~\ref{ass:uniform}.

\begin{theorem}[Scaling with Grid Parameter $m$]\label{th:cbp-m}
With $\ell_i=\alpha m+\beta$, the break probability scales as
\begin{equation}
    \mathrm{CBP}(m) \;\approx\; \operatorname{erfc}\!\left(\frac{k_{\mathrm{eff}}}{\sqrt{2\big((\alpha m+\beta)\sigma_h^2+(\alpha m+\beta-1)\sigma_c^2\big)}}\right).
\end{equation}
\end{theorem}

\begin{proof}
From Theorem~\ref{th:gaussian-tail}, the CBP is
\begin{equation*}
    \mathrm{CBP}(\ell_i) \;\approx\; 
    \operatorname{erfc}\!\left(
    \frac{k_{\mathrm{eff}}}{\sqrt{2\,\mathrm{Var}[\Delta(\ell_i)]}}
    \right).
\end{equation*}
By Theorem~\ref{the:var}, the variance of accumulated error is
\begin{equation*}
    \mathrm{Var}[\Delta(\ell_i)] \;=\; \ell_i\sigma_h^2 + (\ell_i-1)\sigma_c^2.
\end{equation*}
Now substitute the clique embedding relation $\ell_i = \alpha m + \beta$ (Assumption~\ref{ass:uniform}):
\begin{equation}
    \mathrm{Var}[\Delta(\ell_i)] 
    \;=\; (\alpha m + \beta)\sigma_h^2 
          + (\alpha m + \beta - 1)\sigma_c^2.
\end{equation}
Plugging this into the CBP expression yields
\begin{equation*}
    \mathrm{CBP}(m) \;\approx\; 
    \operatorname{erfc}\!\left(
    \frac{k_{\mathrm{eff}}}{\sqrt{2\big((\alpha m+\beta)\sigma_h^2+(\alpha m+\beta-1)\sigma_c^2\big)}}
    \right).
\end{equation*}
\end{proof}

As the grid parameter $m$ (embedding size) increases, the average chain length grows linearly. Longer chains accumulate proportionally more Gaussian noise, which reduces the effective stability margin $k_{\mathrm{eff}}/\sqrt{\mathrm{Var}[\Delta(\ell_i)]}$ and increases the CBP. Thus, embedding overhead directly amplifies noise, linking hardware scaling to reliability degradation.

\subsection{Design Implication}
We derive a design rule for the intra-chain coupling strength $k$ that maintains a target break probability as the embedding grows.

\begin{theorem}[Optimal Scaling of $k$]\label{th:k-scaling}
Let $p^* \in (0,1)$ be a fixed target break probability.  
To maintain $\mathrm{CBP}(\ell_i)=p^*$ as the chain length $\ell_i$ increases, the intra-chain coupling strength must satisfy
\begin{equation}
    k(\ell_i) \;\sim\; \mathcal{O}(\sqrt{\ell_i}).
\end{equation}
In terms of the Zephyr grid parameter $m$, this scaling law can be expressed as
\begin{equation}
    k(m) \;\sim\; \mathcal{O}(\sqrt{m}).
\end{equation}
\end{theorem}

\begin{proof}
From Theorem~\ref{th:gaussian-tail}, the CBP is
\begin{equation*}
    \mathrm{CBP}(\ell_i) \;\approx\; 
    \operatorname{erfc}\!\left(
    \frac{k_{\mathrm{eff}}}{\sqrt{2\,\mathrm{Var}[\Delta(\ell_i)]}}
    \right).
\end{equation*}
Fixing $\mathrm{CBP}(\ell_i)=p^*$ requires that the ratio
remains constant.  
Since $\mathrm{Var}[\Delta(\ell_i)] = \mathcal{O}(\ell_i)$ by~\eqref{the:var_scailing}, 
it follows that $k_{\mathrm{eff}} = \mathcal{O}(\sqrt{\ell_i})$.  
Thus, $k(\ell_i) \sim \mathcal{O}(\sqrt{\ell_i})$, and equivalently $k(m) \sim \mathcal{O}(\sqrt{m})$.
\end{proof}

As the embedding size increases, the accumulated noise in each chain grows linearly with $\ell_i$. To preserve a fixed reliability level, the chain strength only needs to scale as the square root of the chain length (or grid parameter). This establishes a practical design guideline for selecting $k$ in large-scale embeddings.

In summary, our analysis establishes an embedding-aware noise model that explicitly incorporates the growth of chain length in clique embeddings. 
By modeling the accumulated chain-level error as a Gaussian variable with variance increasing linearly in chain length, we derive closed-form laws for CBP and CBF, and show how these scale with the Zephyr grid parameter $m$. 
Furthermore, we provide a design rule for the chain strength $k$, demonstrating that maintaining a constant reliability requires $k$ to scale only as $\sqrt{\ell_i}$ (or $\sqrt{m_i}$). 
This framework links hardware-level control errors to embedding overhead, yielding both predictive accuracy and practical design implications for large-scale QA.



\section{Results}\label{sec:experiments}
This section presents an experimental evaluation of the proposed embedding-aware noise model, including the setup, evaluation metrics, and results across randomly generated QUBO instances. All implementations are made publicly available on our GitHub repository https://github.com/JeongQC/Quantum.

\subsection{Experimental Settings}

\subsubsection{Data settings}
We evaluate the proposed embedding-aware noise model using random QUBO instances of size $L$ (number of logical variables). 
Diagonal terms are sampled from the continuous uniform distribution $\mathcal{U}(-1,1)$, 
and off-diagonal terms are sampled from $\mathcal{U}(-1,1)$ with probability $\rho$ (edge density):
\begin{equation*}
Q(i,i)\sim \mathcal{U}(-1,1),\qquad
Q(i,j)\sim \mathcal{U}(-1,1)\ \text{with probability }\rho \ \ (i<j).
\end{equation*}
Here $\rho$ controls the sparsity of the QUBO graph, where $\rho=1$ produces a fully connected instance. 
The generated QUBO is converted into a binary quadratic model (BQM), the standard quadratic form used in the D-Wave Ocean SDK~\cite{dwave}. 
Since the QPU natively solves ising Hamiltonians with spin variables $s_i\in\{-1,+1\}$, 
we apply the standard QUBO-to-ising mapping $x_i=(1+s_i)/2$ before execution. 
Unless otherwise stated, we set $(h_{\min},h_{\max})=(J_{\min},J_{\max})=(-1,1)$ and use $\rho=1.0$ (dense random QUBO).

\subsubsection{Implementation Details}

All algorithms are implemented in Python~3.8. For QA, we use the official D-Wave Ocean SDK libraries, and all experiments run on the D-Wave \textit{Advantage2\_system1.6} QPU (Zephyr topology). For each $L$, we compute a clique embedding of $K_L$ onto the Zephyr hardware graph using the \textit{minorminer} library, and then fix the found embedding via the \textit{FixedEmbeddingComposite}. 
The maximum embeddable clique size on this solver is $L_{\max}\simeq \,100$ (obtained via \textit{DWaveCliqueSampler}), and we sweep $L$ from 5 up to $L_{\max}$ in increments of~5. We run an annealing time sweep $T\in\{5,20,100,200\}\,\mu$s with a linear annealing schedule. For reproducibility, unless stated otherwise: QUBO seed $= L$ and embedding seed $= L$. 

To compare the solution quality, we adopt various classical baselines. As classical baselines, we implement simulated annealing (SA)~\cite{10907925, kirkpatrick1983optimization} using the \textit{dwave-neal} library, and mixed-integer linear program (MILP) approaches using PuLP (open-source classical solver)~\cite{santos2020mixed} and Gurobi (commercial MILP solver)~\cite{gurobi}. For SA, we sample $2000$ states with a logarithmic temperature schedule, and report the lowest obtained energy and runtime. For PuLP and Gurobi, we directly formulate the QUBO as an MILP and report the optimal objective values and wall-clock time. This setup enables a consistent comparison across QA, SA, PuLP, and Gurobi in terms of solution quality and runtime.

\textbf{Fitting}:
The goal of fitting is to calibrate the embedding-aware noise model to the observed hardware behavior. 
Specifically, we estimate the parameters $(\sigma_h,\sigma_c,\kappa)$, which represent the variance of local field errors, 
the variance of coupler errors, and the effective stabilizing margin ($\kappa=\eta k$). 
These parameters enter the analytical expression for the CBP of a chain of length~$\ell_i$:
\begin{equation*}
\mathrm{CBP}(\ell_i) \;\approx\; 
\operatorname{erfc}\!\left(
\frac{\kappa}{\sqrt{2\left(\ell_i\sigma_h^2+(\ell_i-1)\sigma_c^2\right)}}
\right).    
\end{equation*}
For an embedding with chain lengths $\{\ell_i\}$, the model-predicted CBF is then 
\begin{equation}
\overline{\mathrm{CBF}}_{\mathrm{pred}}(L) 
= \frac{1}{|V_L|}\sum_{i=1}^{|V_L|}\mathrm{CBP}(\ell_i).    
\end{equation}
We fit $(\sigma_h,\sigma_c,\kappa)$ by minimizing the sum of squared errors (SSEs) between observed and predicted mean CBF:
\begin{equation}
\text{SSE}(\sigma_h,\sigma_c,\kappa) 
= \sum_{L}\left(\overline{\mathrm{CBF}}_{\mathrm{obs}}(L) 
- \overline{\mathrm{CBF}}_{\mathrm{pred}}(L;\sigma_h,\sigma_c,\kappa)\right)^2.    
\end{equation}

The search ranges are motivated by prior ICE characterizations on D-Wave processors, 
which report typical control error magnitudes of about $|\delta h_i|\!\approx\!0.05$ 
and $|\delta J_{ij}|\!\approx\!0.02$ in Ising units~\cite{dwave, yarkoni2022quantum}. 
Given that our QUBO coefficients are drawn from $[-1,1]$, these values correspond to noise on the order of a few percent of the programmed range. 
Accordingly, we select $\sigma_h,\sigma_c\in[0.005,0.08]$, spanning roughly 0.5-8\% of the coefficient scale 
and thus covering both weaker and stronger noise scenarios observed in practice. 
For the effective stabilization margin, we search $\kappa\in[0.10,1.00]$, which is consistent with practical intra-chain coupling strengths on Advantage-class hardware. 
The step sizes (0.005 for $\sigma_h,\sigma_c$ and 0.05 for $\kappa$) balance resolution with computational tractability.


\subsection{Evaluation Metrics}
We use the following metrics to evaluate embedding-aware noise behavior and model fit:
\begin{itemize}
    \item CBF:  
    for each read, we obtain a CBF $\operatorname{cbf}\in[0,1]$ from the sampler.  
    Over $N$ reads, the observed mean is defined as 
    \begin{equation}
    \overline{\mathrm{CBF}}_{\mathrm{obs}}
    =\frac{1}{N}\sum_{n=1}^N \operatorname{cbf}_n.
    \end{equation} 

    \item Energy statistics:  
    from the Hamiltonian energies $\{E_n\}$, we report the best (minimum) energy for each problem size $L$ (number of logical variables) and annealing time $T$ (in $\mu$s), i.e., $E_{\min}$. 

    \item Fit quality:
    given fitted $(\sigma_h,\sigma_c,\kappa)$, we compute predicted mean model-predicted CBF ($\overline{\mathrm{CBF}}_{\mathrm{pred}}(L)$) 
    and evaluate the SSE:
    \[
    \text{SSE}=\sum_L\big(\overline{\mathrm{CBF}}_{\mathrm{obs}}(L)-\overline{\mathrm{CBF}}_{\mathrm{pred}}(L)\big)^2.
    \]
    We also visualize observed vs.\ predicted curves and per-$L$ absolute error.
\end{itemize}

\subsection{Research Questions}
We structure the experiments as follows. Here $L$ denotes the number of logical variables (clique size), $T$ is the annealing time (in $\mu$s), and $k$ represents the chain strength.

\begin{enumerate}
    \item[RQ1:] Chain length scaling under clique embeddings.  
    The first question is whether the average chain length $\bar\ell$ grows linearly with $L$. 
    Confirming this relationship validates the structural basis of the embedding-aware noise model, 
    since longer chains imply larger noise accumulation.

    \item[RQ2:] Embedding scaling of CBF at fixed $T$ and $k$.  
    Given that $\bar\ell$ increases linearly with $L$, the next question is whether the observed mean CBF 
    also increases accordingly, consistent with the variance growth predicted by the model. 
    This links structural scaling to the growth of CBP. 
    We examine observed $\overline{\mathrm{CBF}}_{\mathrm{obs}}(L)$ together with the predicted 
    $\overline{\mathrm{CBF}}_{\mathrm{pred}}(L)$ obtained from the fitted model parameters.

    \item[RQ3:] Schedule sensitivity when varying $T$ at fixed $k$.  
    Here, we ask how different annealing times $T\in\{5,20,100,200\}\,\mu$s influence the CBF scaling with $L$. 
    The aim is to determine whether the observed trends are robust across schedules and whether 
    the fitted parameters $(\sigma_h,\sigma_c,\kappa)$ remain stable across $T$.

    \item[RQ4:] Chain strength sensitivity when varying $k$ at fixed $T$.  
    This question concerns whether sweeping the chain strength $k\in [0.1,2.5]$ reproduces 
    the predicted $\operatorname{erfc}$ tails and the sublinear rule $k=\mathcal{O}(\sqrt{\bar{\ell}})$. 
    The purpose is to test whether modest increases in $k$ are sufficient to stabilize chains 
    without overly suppressing logical couplers.

    \item[RQ5:] Algorithmic comparison with classical baselines. In this question, we benchmark QA against SA, MILP, and Gurobi. All methods are applied to the same randomly generated QUBOs, and we compare the resulting energy quality and runtime scalability as $L$ increases. This analysis highlights whether QA provides competitive or complementary advantages relative to classical heuristics and exact solvers under embedding-aware settings.

\end{enumerate}

\subsection{RQ1: Embedding scaling (chain length vs number of logical variables)}

A central assumption of our embedding-aware noise model is that in clique embeddings, the average chain length $\bar\ell$ increases approximately linearly with the number of logical variables $L$. 
To confirm this structural property, we generated random QUBO instances of varying size and embedded them onto the 
D-Wave \textit{Advantage2\_system1.6} QPU using clique embeddings. 
\begin{figure}[!h]
    \centering
    \includegraphics[width=0.55\linewidth]{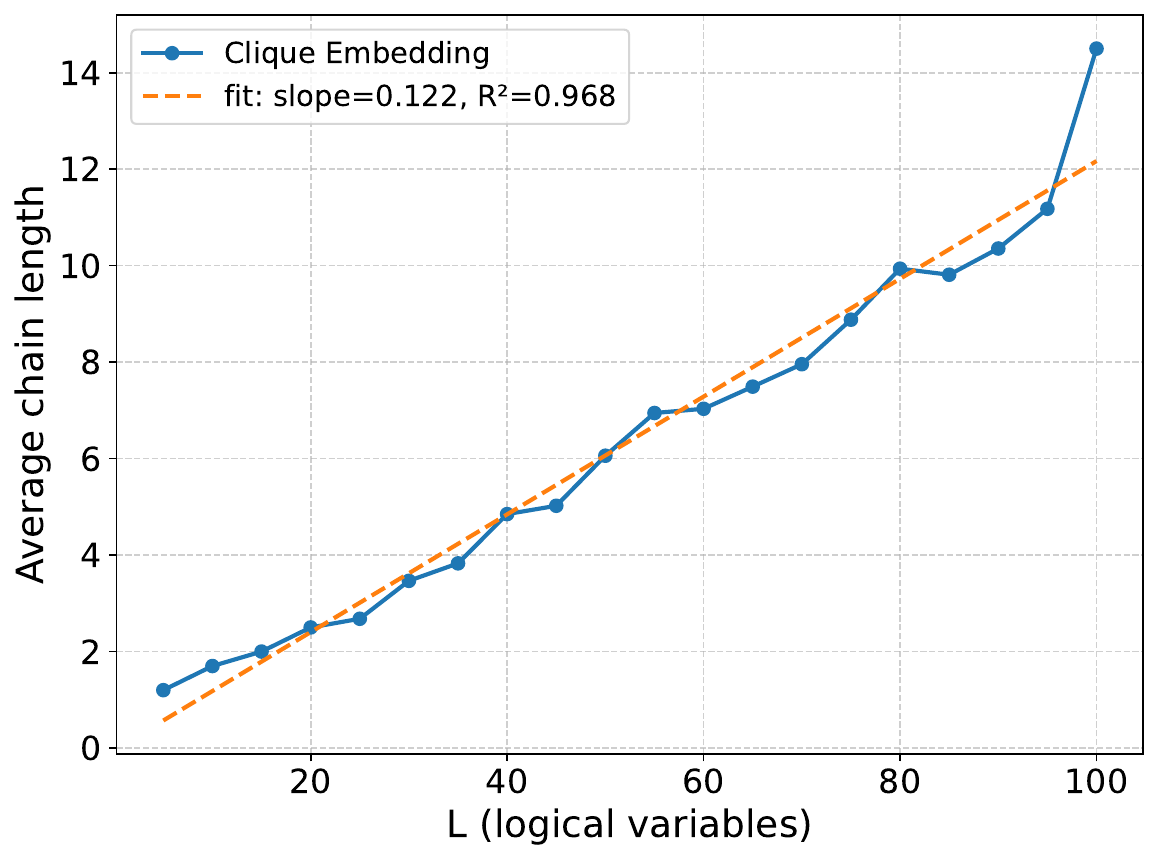}
    \caption{Average chain length as a function of logical problem size $L$ under clique embeddings on the D-Wave \textit{Advantage2\_system1.6} QPU.}
    \label{fig:chainlen}
\end{figure}

Figure~\ref{fig:chainlen} shows the observed average chain length as a function of $L$ together with a linear fit. 
The fitted line has slope $0.122$ with $R^2=0.968$, indicating an almost perfectly linear trend; equivalently, $\bar\ell$ increases by about $0.122$ per additional logical variable (roughly one unit per $\sim 8.2$ variables). 
This empirical observation validates the structural basis of our theory and provides the foundation for linking chain length growth to noise accumulation and, ultimately, to increasing CBFs.

\subsection{RQ2: Observed vs.\ predicted CBF at fixed settings ($T=5\,\mu$s, $k=1.0$)}

\begin{figure}[!h]
    \centering
    \includegraphics[width=0.55\linewidth]{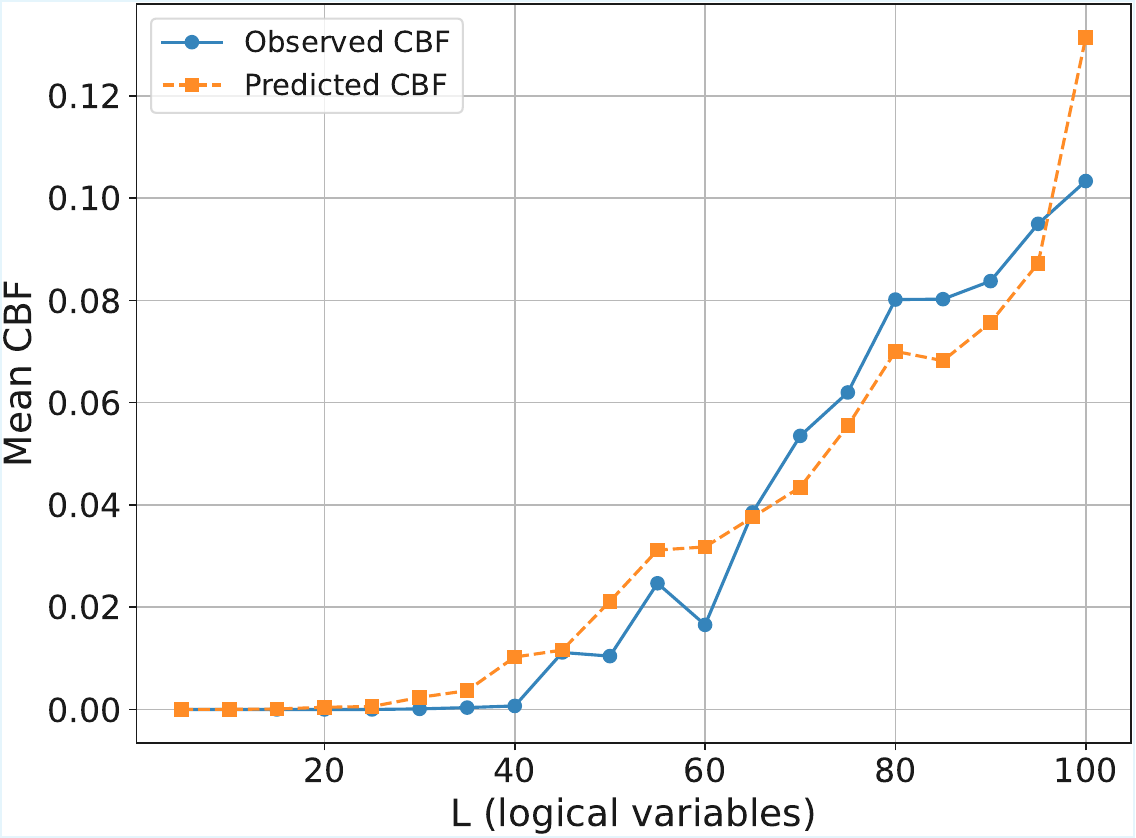}
    \caption{Observed mean CBF (blue) versus model-predicted CBF (orange) as a function of logical problem size $L$ 
    under clique embeddings on the D-Wave \textit{Advantage2\_system1.6} QPU. 
    Results are shown for annealing time $T=5\,\mu$s and chain strength $k=1.0$. 
    }
    \label{fig:cbf_obs_pred}
\end{figure}

\begin{table}[h]
\centering
\caption{Observed vs.\ predicted mean CBF at $T=5\,\mu$s and $k=1.0$. 
Here $|\Delta|$ denotes the absolute error 
$|\overline{\mathrm{CBF}}_{\text{obs}} - \overline{\mathrm{CBF}}_{\text{pred}}|$.}
\label{tab:fit_obs_pred}
\begin{tabular}{rccc}
\hline
$L$ & $\overline{\mathrm{CBF}}_{\text{obs}}$ & $\overline{\mathrm{CBF}}_{\text{pred}}$ & $|\Delta|$ \\
\hline
5   & 0.00000   & 5.3e-06   & 5.3e-06 \\
10  & 0.00000   & 1.9e-05   & 1.9e-05 \\
15  & 0.00000   & 6.5e-05   & 6.5e-05 \\
20  & 0.00000   & 3.3e-04   & 3.3e-04 \\
25  & 0.00048   & 4.36e-04  & 4.0e-05 \\
30  & 0.00013   & 1.85e-03  & 0.00172 \\
35  & 0.00000   & 2.92e-03  & 0.00292 \\
40  & 0.00076   & 8.44e-03  & 0.00768 \\
45  & 0.00690   & 9.62e-03  & 0.00272 \\
50  & 0.01056   & 0.01787   & 0.00731 \\
55  & 0.01454   & 0.02678   & 0.01224 \\
60  & 0.01411   & 0.02728   & 0.01316 \\
65  & 0.03384   & 0.03253   & 0.00131 \\
70  & 0.05027   & 0.03781   & 0.01246 \\
75  & 0.05440   & 0.04875   & 0.00564 \\
80  & 0.07276   & 0.06220   & 0.01056 \\
85  & 0.07718   & 0.06046   & 0.01672 \\
90  & 0.07351   & 0.06738   & 0.00613 \\
95  & 0.09064   & 0.07822   & 0.01242 \\
100 & 0.08680   & 0.12016   & 0.03336 \\
\hline
\end{tabular}
\end{table}

Figure~\ref{fig:cbf_obs_pred} and Table~\ref{tab:fit_obs_pred} show that the observed mean CBF increases consistently with problem size $L$, reflecting the accumulation of control errors as chain lengths grow. 
The embedding-aware noise model, fitted via $(\sigma_h,\sigma_c,\kappa)$, yields predicted CBF values that closely match the experimental observations across the full range of $L$. 
The best-fit parameters were $\sigma_h=0.06$, $\sigma_c=0.005$, and $\kappa=0.35$, with a $\text{SSE}=0.0021$. 
The absolute error $|\Delta| = |\overline{\mathrm{CBF}}_{\text{obs}} - \overline{\mathrm{CBF}}_{\text{pred}}|$ remains below $0.03$ for most instances, with only moderate deviations at the largest $L$ where chains are longest. 
These results confirm that a Gaussian misspecification model, combined with measured chain length distributions from clique embeddings, is sufficient to capture the dominant trend of CBF growth on the Advantage2 hardware.

\subsection{RQ3: Schedule sensitivity (vary $T$, fix $k$)}\label{res:RQ3}

\begin{figure}[!h]
\centering
\subfigure[$T=5\,\mu$s.]{%
\includegraphics[width=0.45\textwidth]{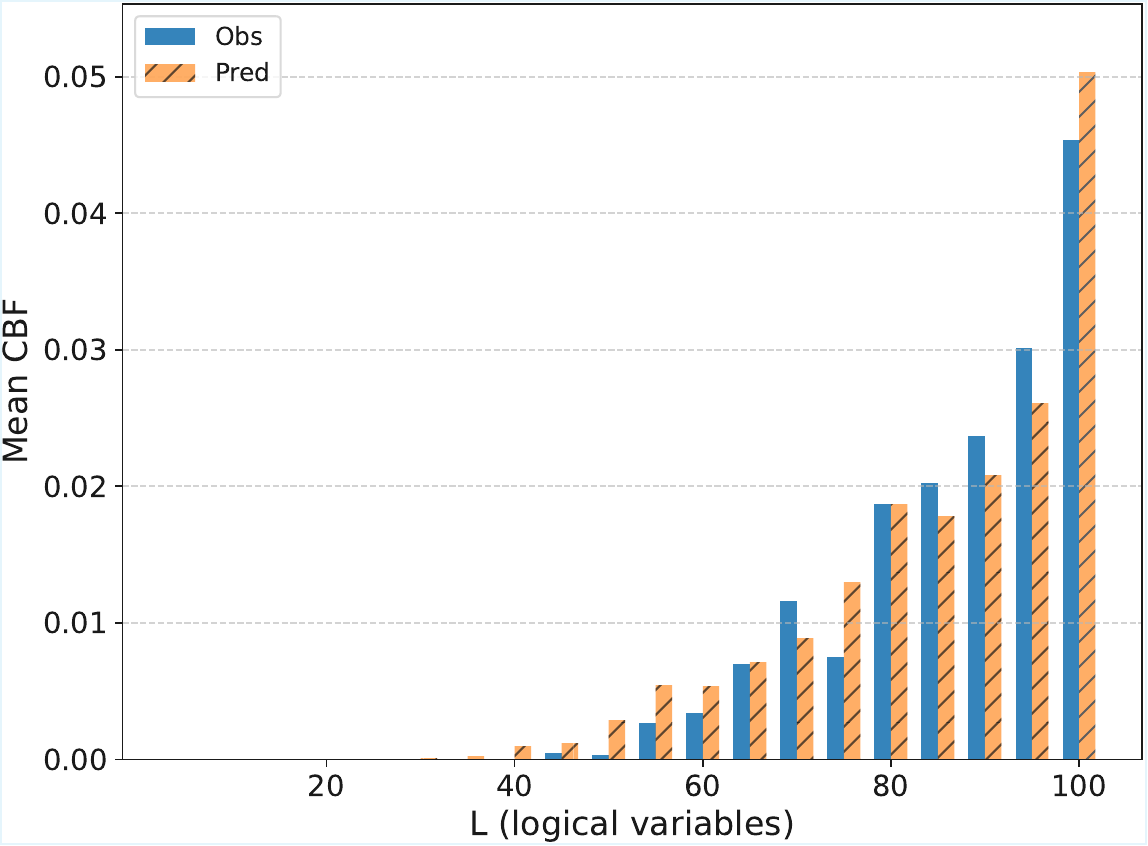}%
    \label{fig.predobs5}}
\subfigure[$T=20\,\mu$s.]{%
\includegraphics[width=0.45\textwidth]{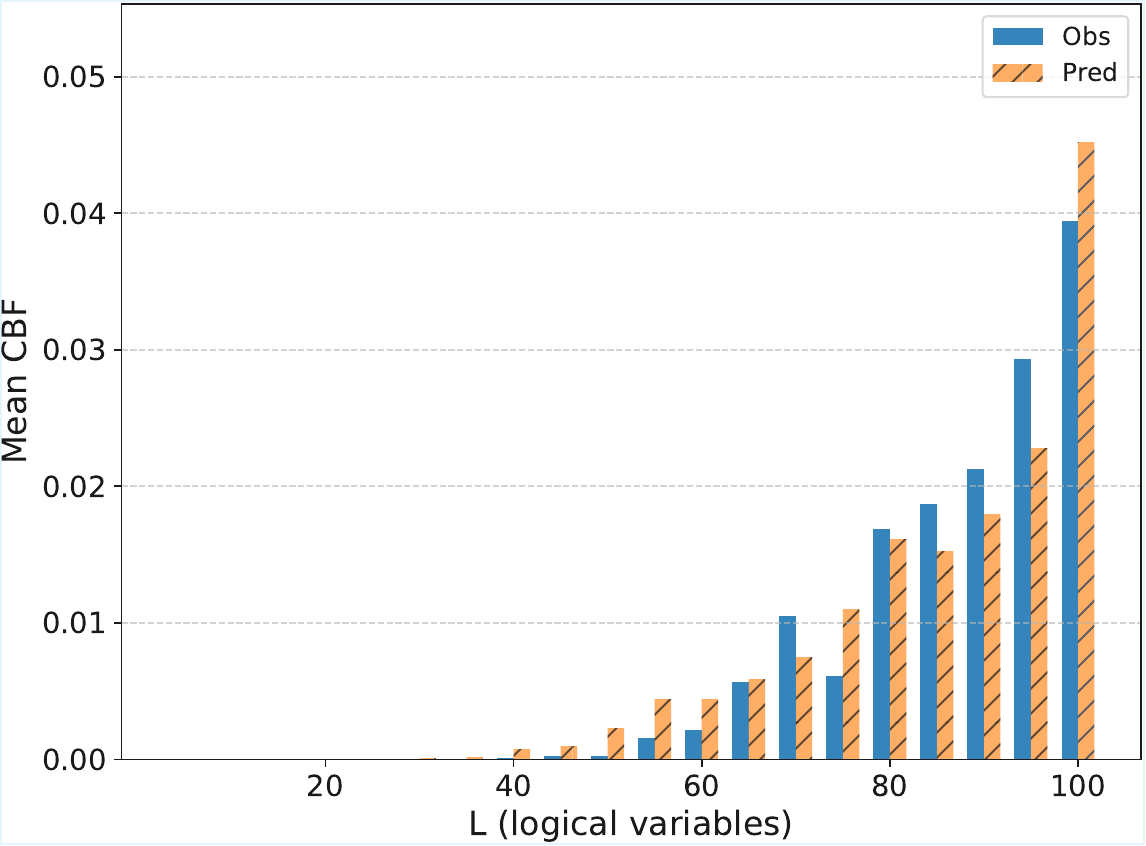}%
    \label{fig.predobs20}}
\subfigure[$T=100\,\mu$s.]{%
\includegraphics[width=0.45\textwidth]{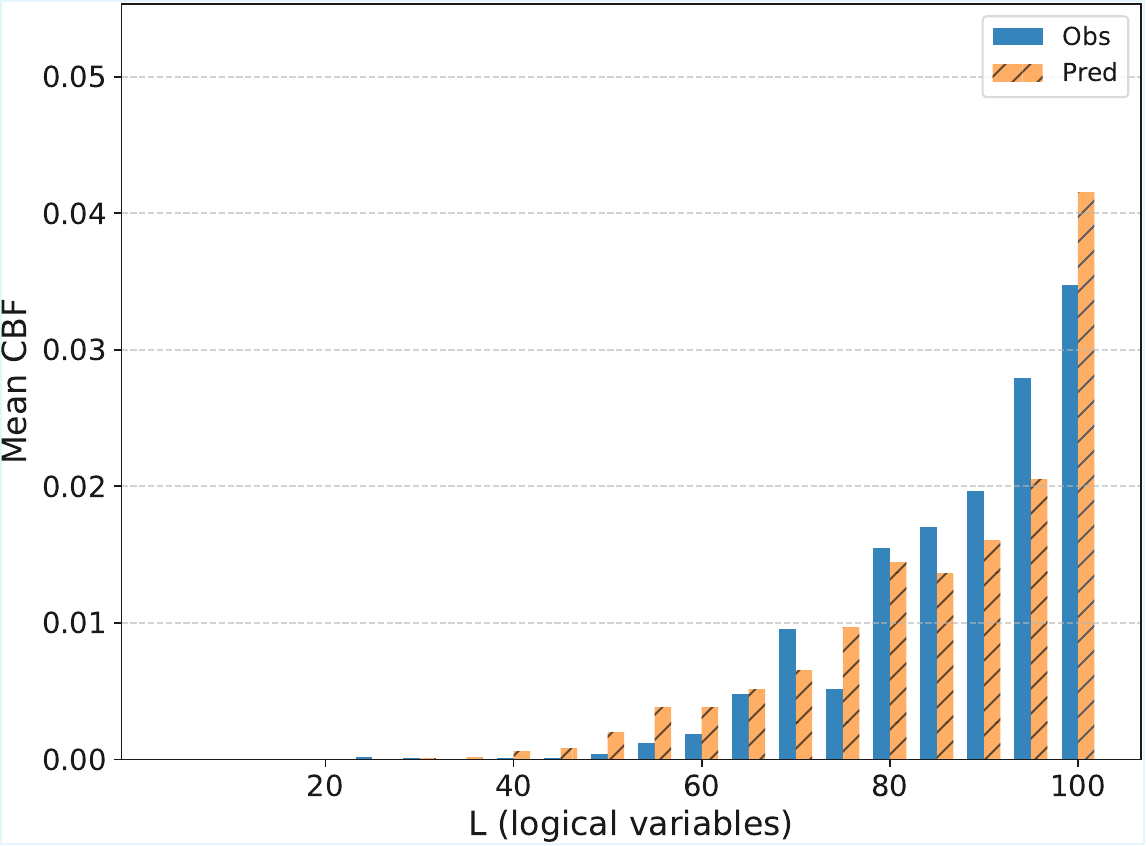}%
    \label{fig.predobs100}}
\subfigure[$T=200\,\mu$s.]{%
\includegraphics[width=0.45\textwidth]{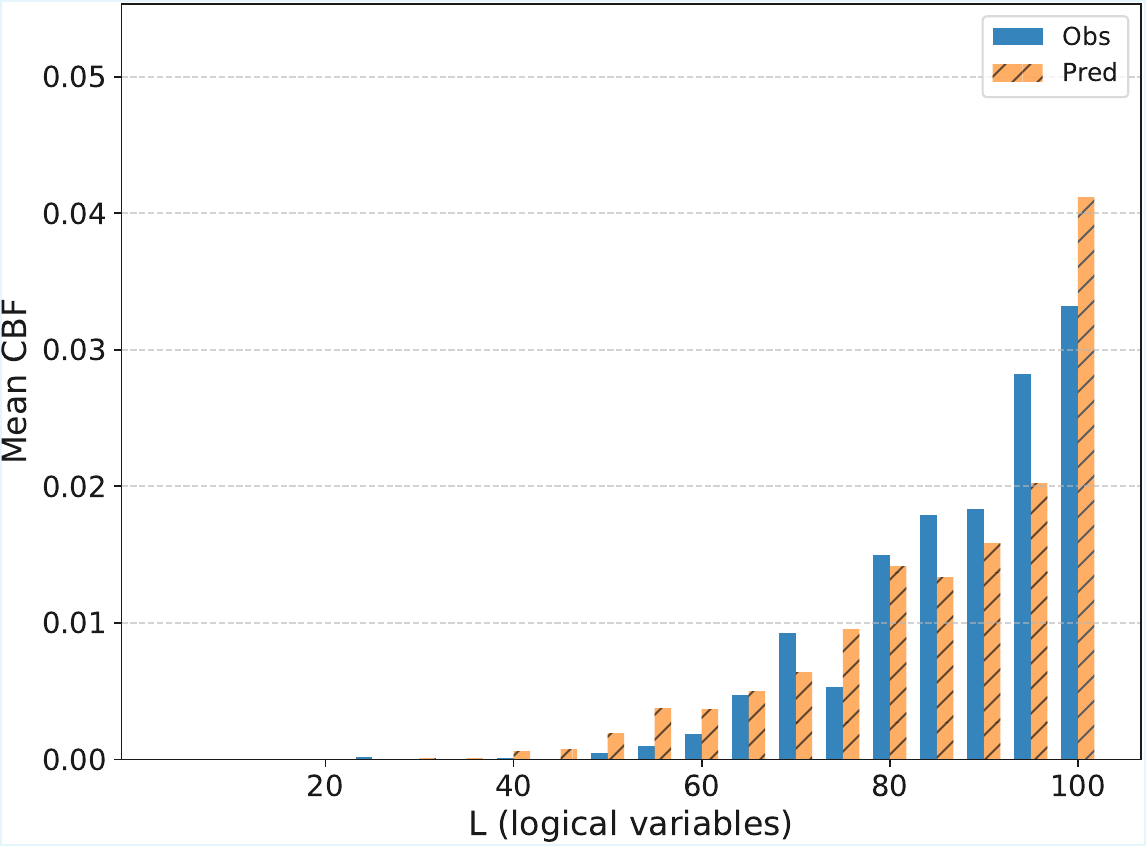}%
    \label{fig.predobs200}}
    
\caption{Observed (blue bars) and model-predicted (orange hatched bars) mean CBF as a function of problem size $L$ under clique embeddings on the D-Wave \textit{Advantage2\_system1.6} QPU. 
Each subfigures correspond to annealing times $T\in\{5,20,100,200\}\,\mu$s at fixed chain strength $k=1.5$.
}
\label{fig:cbf_obs_pred1}
\end{figure}

Figure~\ref{fig:cbf_obs_pred1} compares observed and predicted mean CBF as a function of $L$ 
for annealing times $T=5,20,100,200\,\mu$s at fixed chain strength $k=1.5$. 
In all cases, the observed CBF curves increase with $L$, consistent with the accumulation of control errors as chain lengths grow. 
At the same time, the absolute level of CBF decreases as the annealing time becomes longer, indicating that longer schedules mitigate chain breaks by allowing more adiabatic evolution. 
Across all schedules, the predicted curves closely follow the observed data, indicating that the embedding-aware noise model remains stable under different annealing times.

\begin{table}[h]
\centering
\caption{Fitted parameters $(\sigma_h,\sigma_c,\kappa)$ and fit error (SSE) for different annealing times $T$ at fixed $k=1.5$.}
\label{tab:fit_params_T}
\begin{tabular}{cccccc}
\hline
$T$ [$\mu$s] & $\sigma_h$ & $\sigma_c$ & $\kappa$ & SSE \\
\hline
5   & 0.065 & 0.015 & 0.50 & $1.1\times 10^{-4}$ \\
20  & 0.070 & 0.015 & 0.55 & $1.5\times 10^{-4}$ \\
100 & 0.070 & 0.005 & 0.55 & $1.7\times 10^{-4}$ \\
200 & 0.050 & 0.010 & 0.40 & $1.9\times 10^{-4}$ \\
\hline
\end{tabular}
\end{table}

Table~\ref{tab:fit_params_T} shows that the fitted parameters vary only modestly with $T$. 
At $T=5\,\mu$s we obtain $\sigma_h=0.065$, $\sigma_c=0.015$, $\kappa=0.50$ with $\text{SSE}=1.1\times 10^{-4}$. 
At $T=20\,\mu$s the fitted values are $\sigma_h=0.070$, $\sigma_c=0.015$, $\kappa=0.55$ with $\text{SSE}=1.5\times 10^{-4}$. 
At $T=100\,\mu$s we find $\sigma_h=0.070$, $\sigma_c=0.005$, $\kappa=0.55$ with $\text{SSE}=1.7\times 10^{-4}$. 
Finally, at $T=200\,\mu$s the fit gives $\sigma_h=0.050$, $\sigma_c=0.010$, $\kappa=0.40$ with $\text{SSE}=1.9\times 10^{-4}$. Here, we denote the fitted $k_{\mathrm{eff}}$ by $\kappa$.
These results show that $(\sigma_h,\sigma_c,\kappa)$ remain within a narrow range across different schedules, confirming that control errors dominate chain breaks while the overall CBF level is modulated by the annealing time. 
This demonstrates that the embedding-aware Gaussian noise model is robust across schedule variations and captures both the scaling with $L$ and the suppression of chain breaks with longer anneals.

\subsection{RQ4: Chain strength sensitivity (vary $k$, fix $T$)}\label{res:RQ4}

\begin{figure}[!h]
    \centering
    \includegraphics[width=0.55\linewidth]{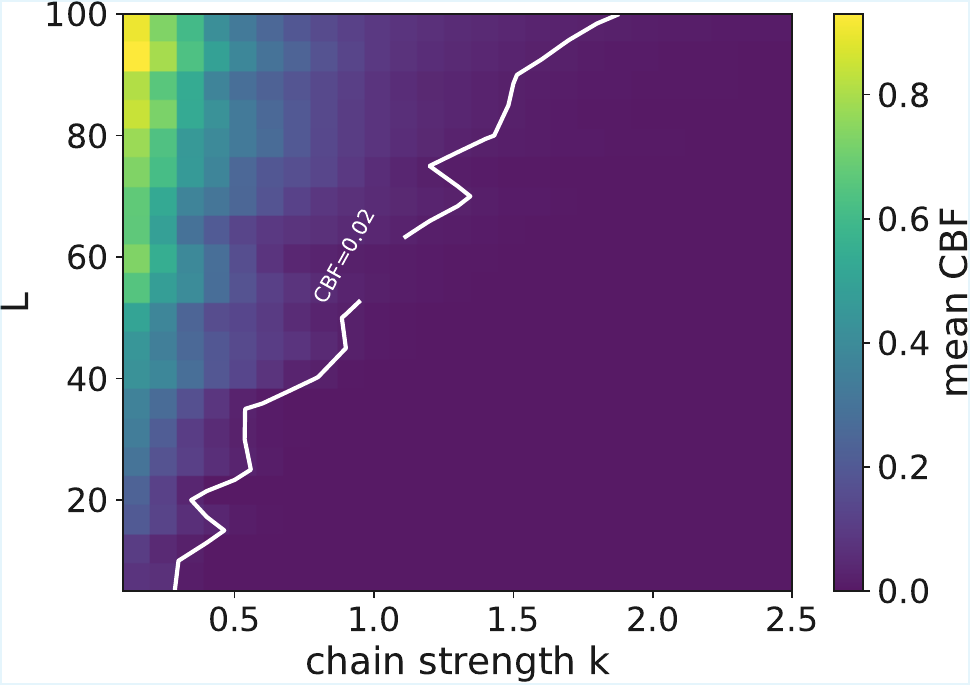}
    \caption{Heatmap of observed mean CBF as a function of problem size $L$ and chain strength $k$ at fixed annealing time $T=20\,\mu$s.}
    \label{fig:cbf_heat}
\end{figure}

Figure~\ref{fig:cbf_heat} presents the observed mean CBF across problem sizes $L$ and chain strengths $k$ at fixed annealing time $T=20\,\mu$s. 
The CBF$=0.02$ contour delineates the empirical stability boundary, $k^\ast(L)$, which increases systematically with $L$. 
This indicates that larger chain strengths are required to maintain chain integrity as the average chain length grows, consistent with the intuition that accumulated control errors scale with the embedding overhead.

\begin{figure}[!h]
\centering
\subfigure[$\tau=0.01$.]{%
    \includegraphics[width=0.32\textwidth]{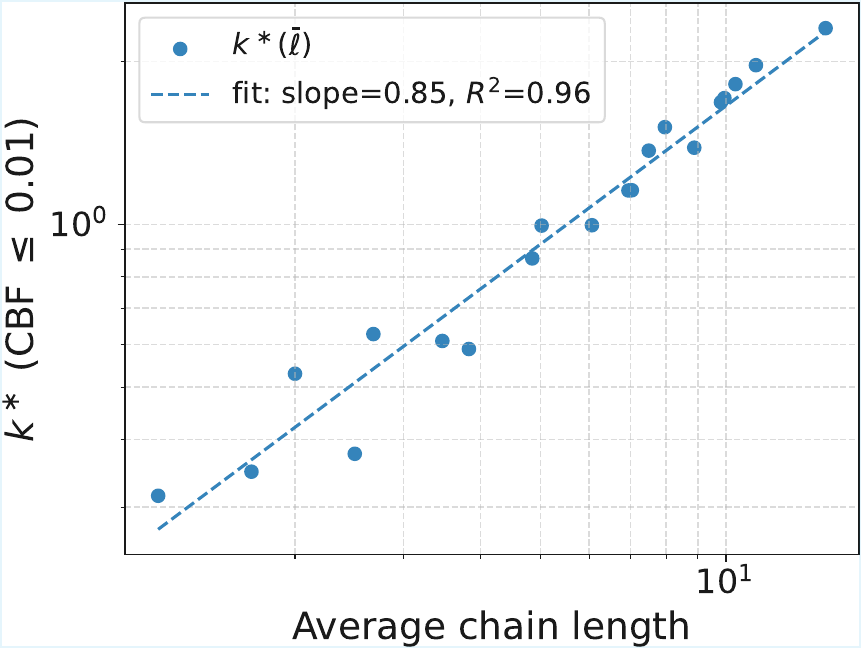}%
    \label{fig.kstar_tau001}}
\subfigure[$\tau=0.02$.]{%
    \includegraphics[width=0.32\textwidth]{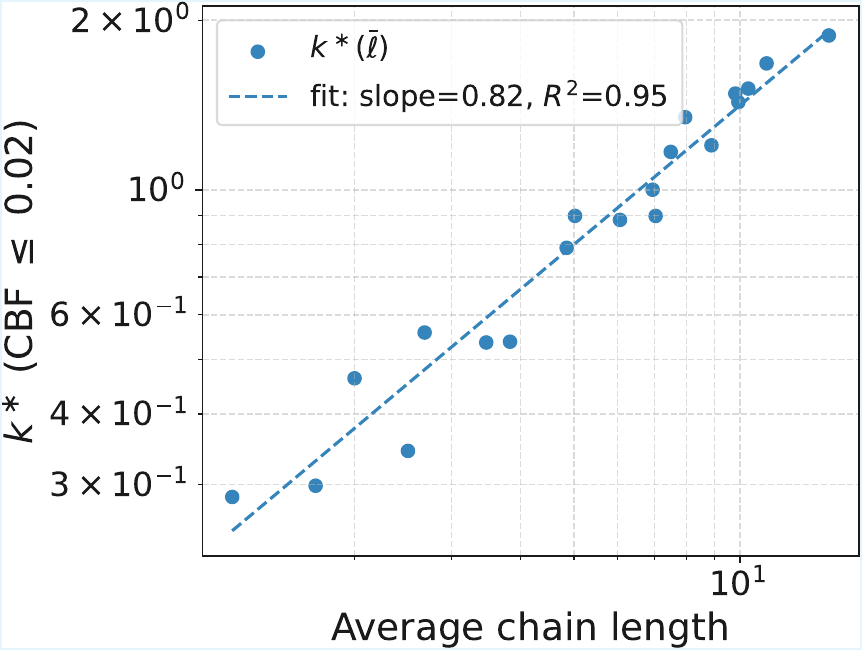}%
    \label{fig.kstar_tau002}}
\subfigure[$\tau=0.05$.]{%
    \includegraphics[width=0.32\textwidth]{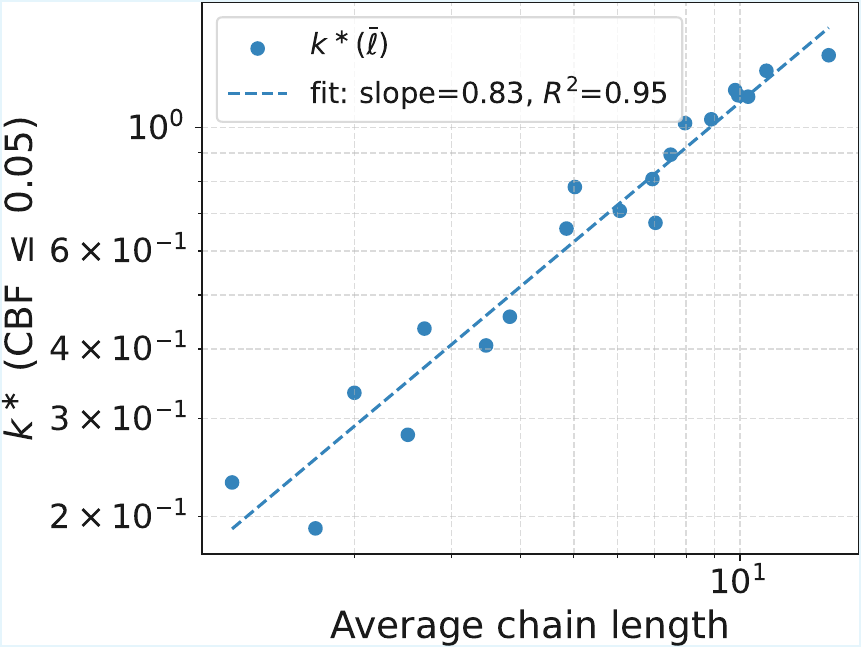}%
    \label{fig.kstar_tau005}}
\caption{Log-log plots of the critical chain strength $k^\ast(\bar{\ell})$ at thresholds $\tau\in\{0.01,0.02,0.05\}$.}
\label{fig:kstar_scaling}
\end{figure}

To quantify this scaling, we extract $k^\ast(\bar{\ell})$ as a function of the average chain length $\bar{\ell}$ induced by clique embeddings (with $L$ logical variables), under thresholds $\tau\in\{0.01,0.02,0.05\}$, and plot the results on log-log axes in Figure~\ref{fig:kstar_scaling}.
In all three cases, $k^\ast(\bar{\ell})$ follows a clear power-law dependence on the average chain length, with fitted exponents $\alpha\simeq 0.82$-$0.85$ and coefficients of determination $R^2\approx 0.95$. 
These exponents are substantially steeper than the theoretical $\alpha=0.5$ predicted by the $\sqrt{\bar{\ell}}$ law under independent Gaussian noise assumptions. 
This deviation indicates that practical noise on the D-Wave \textit{Advantage2\_system1.6} QPU may involve correlations or additional amplification mechanisms that accelerate the growth of the required chain strength with embedding size. 
Nevertheless, the scaling remains sublinear ($\alpha<1$), confirming that chain strength does not need to increase linearly with chain length to ensure stability. 
We defer a more detailed discussion of these deviations, including possible correlated noise sources and their implications for embedding-aware calibration, to Section~\ref{sec:discussion}.

\begin{figure}[!h]
    \centering
    \includegraphics[width=0.55\linewidth]{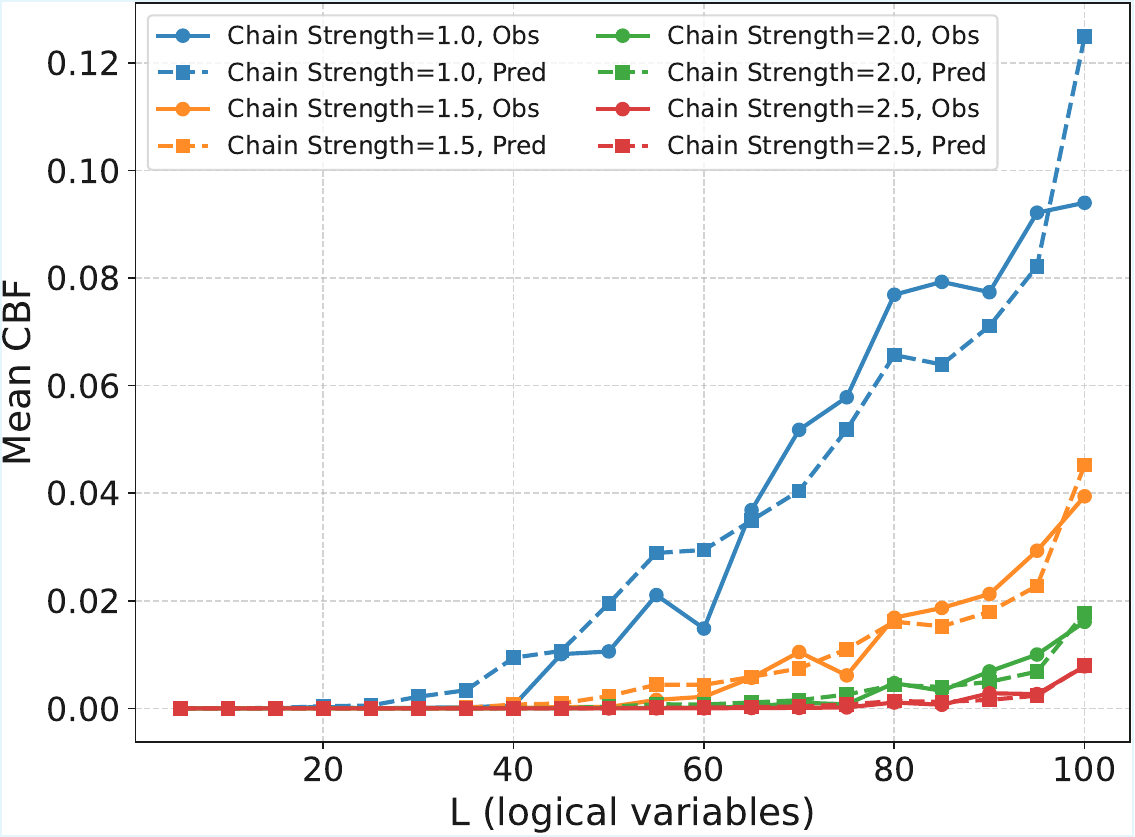}
    \caption{Observed (solid lines) and model-predicted (dashed lines) mean CBF under clique embeddings on the D-Wave \textit{Advantage2\_system1.6} QPU at $T=20\,\mu$s.}
    \label{fig:cbf_obs_pred2}
\end{figure}

Figure~\ref{fig:cbf_obs_pred2} compares observed and predicted mean CBF curves for representative chain strengths $k\in\{1.0,1.5,2.0,2.5\}$. 
As $k$ increases, the observed CBF decreases monotonically across all $L$, confirming that stronger chain strength suppresses chain breaks by enlarging the effective stabilizing margin $k_{\mathrm{eff}}$. 
Concurrently, the model predictions align more closely with experimental data at larger $k$, as reflected in the steadily decreasing SSE values reported in Table~\ref{tab:fit_params_k}. 
This demonstrates that the embedding-aware Gaussian noise model captures not only the dependence of CBF on $L$, but also the modulation introduced by varying chain strength.

\begin{table}[!h]
\centering
\caption{Fitted parameters $(\sigma_h,\sigma_c,\kappa)$ and fit error (SSE) for different chain strengths $k$ at fixed $T=20\,\mu$s.}
\label{tab:fit_params_k}
\scriptsize
\begin{tabular}{ccccc}
\hline
$k$ & $\sigma_h$ & $\sigma_c$ & $\kappa$ & SSE \\
\hline
1.0 & 0.06 & 0.005 & 0.35 & 0.00207 \\
1.5 & 0.07 & 0.015 & 0.55 & 0.00015 \\
2.0 & 0.08 & 0.015 & 0.75 & 0.000022 \\
2.5 & 0.08 & 0.045 & 0.95 & 0.0000024 \\
\hline
\end{tabular}
\end{table}

Overall, these results provide a comprehensive characterization of chain strength sensitivity. 
The empirical scaling of $k^\ast(\bar{\ell})$ establishes a practical design rule: while the theoretical $\sqrt{\bar{\ell}}$ law provides a lower bound, real hardware requires a steeper, yet still sublinear, growth in chain strength. 
Furthermore, the improved concordance between model predictions and experimental data at larger $k$ underscores both the utility and the limitations of the Gaussian noise assumption. 
These findings highlight the necessity of embedding-aware calibration strategies when applying QA to large-scale QUBO instances.

\subsection{RQ5: Algorithmic Comparison with Classical Baselines}

\begin{table}[h]
\centering
\scriptsize
\caption{Comparison of best energy ($E_{\min}$) and computation time ($T_{c}$) across different solvers (SA, PuLP, Gurobi, QA). For QA, the time is computed as read counts multiplied by annealing time.}
\label{tab:energy_time_comparison}
\begin{tabular}{ccccccccc}
\hline
\multirow{2}{*}{L} & \multicolumn{2}{c}{SA}         & \multicolumn{2}{c}{PuLP}       & \multicolumn{2}{c}{Gurobi}     & \multicolumn{2}{c}{QA}         \\
                   & $E_{\min}$ & $T_{c}$ & $E_{\min}$ & $T_{c}$ & $E_{\min}$ & $T_{c}$ & $E_{\min}$ & $T_{c}$ \\ \hline
5   & -0.9622   & 0.10  & -0.9622   & 0.02   &  -0.9622     &  0.00         & -0.9622   & 0.01 \\
10  & -4.1441   & 0.28  & -4.1441   & 0.02   &  -4.1441     &  0.00         & -4.1441   & 0.01 \\
15  & -4.6934   & 0.44  & -4.6934   & 0.42   &  -4.6934     &  0.01         & -4.6934   & 0.01 \\
20  & -21.9189  & 0.57  & -21.9189  & 0.87   &  -21.9189    &  0.02         & -21.9189  & 0.01 \\
25  & -20.2234  & 0.74  & -20.2234  & 2.19   &  -20.2234    &  0.06         & -20.2234  & 0.01 \\
30  & -27.6509  & 0.92  & -27.6509  & 7.32   &  -27.6509    &  0.12         & -27.6509  & 0.01 \\
35  & -44.6379  & 1.11  & -44.6379  & 9.54   &  -44.6379    &  0.25         & -44.6379  & 0.01 \\
40  & -49.8143  & 1.35  & -49.8143  & 132.61 &  -49.8143    &  0.82         & -49.8143  & 0.01 \\
45  & -55.6170  & 1.58  & -55.6170  & 720.72 &  -55.6170    &  0.05         & -55.6170  & 0.01 \\
50  & -70.2638  & 1.80  & -70.2638  & 2788.69&  -70.2638    &  0.06         & -70.2638  & 0.01 \\
55  & -64.1037  & 2.12  & -42.3855  & 3600.32&  -64.1037    &  0.11         & -64.1037  & 0.01 \\
60  & -97.5980  & 2.24  & -69.5739  & 3600.27&  -97.5980    &  0.05         & -97.5980  & 0.01 \\
65  & -111.3206 & 2.44  & -62.6372  & 3600.37&  -111.3206   &  0.13         & -111.3206 & 0.01 \\
70  & -115.4876 & 2.85  & -65.3950  & 3600.31&  -115.4876   &  3600.10        & -115.4876 & 0.01 \\
75  & -102.3171 & 3.24  & -48.8348  & 3600.31&  -102.3171   &  3600.07        & -102.3171 & 0.04 \\
80  & -125.3219 & 3.38  & -51.3503  & 3600.39&  -125.3219   &  3600.10        & -125.3219 & 0.04 \\
85  & -100.4712 & 3.93  & -38.9299  & 3600.32&  -100.4712   &  3600.10        & -100.4712 & 0.20 \\
90  & -146.3845 & 3.99  & -73.3973  & 3600.30&  -146.3845   &  3600.10        & -146.3845 & 0.20 \\
95  & -127.0726 & 4.64  & -46.9583  & 3600.38&  -127.0726   &  3600.07        & -127.0726 & 0.20 \\
100 & -159.2677 & 4.76  & -72.2637  & 3600.51&  -159.2677   &  3600.10        & -159.2677 & 0.20 \\ \hline
\end{tabular}
\end{table}

Table~\ref{tab:energy_time_comparison} compares the best energy values and computation times for SA, PuLP, Gurobi, and QA. 
For QA, the runtime is reported as the product of the number of reads ($2000$) and the programmed annealing time. For SA, the number of reads was fixed at $2000$ across all problem sizes $L$. For PuLP and Gurobi, a computation time limit of $3600$ seconds was imposed.

Across all $L$, the best energy values obtained by QA, SA, and Gurobi remain essentially identical. In contrast, PuLP begins to deviate from the optimum beyond $L = 55$, returning higher-energy solutions due to hitting the $3600$~s time limit. This behavior highlights the limitation of PuLP’s branch-and-bound strategy when applied to dense QUBO formulations, where the solver can terminate with a feasible but suboptimal solution if the search tree cannot be fully explored in time. 

The runtime behavior diverges significantly across solvers. The SA shows a relatively mild growth in runtime because the number of reads is fixed due to our simulation setting, and the heuristic nature of the algorithm keeps the scaling modest even as the QUBO size increases. However, the PuLP and Gurobi quickly hit the $3600$~s limit as $L$ grows, with PuLP failing to return optimal solutions beyond $L = 55$ and Gurobi stalling beyond $L = 70$. Although Gurobi’s reported energies coincide with the reference values up to $L=65$, its solver logs do not certify optimality within the $3600$~s time limit for $L \geq 70$, reflecting the intrinsic hardness of QUBO as NP-hard~\cite{glover2022quantum}. QA exhibits an almost constant runtime of $10$-$200$~ms across all $L$, since execution time is determined solely by the programmed anneal time and the fixed number of reads. 
The best energies obtained from QA match those reported by SA and Gurobi, indicating that QA achieves competitive solution quality within a fixed and predictable computation time.

\section{Discussion}\label{sec:discussion}
The preceding results demonstrate how intrinsic control errors manifest in the observed chain break statistics and how our embedding-aware Gaussian model captures these effects under variations in annealing time and chain strength. 
In this section, we interpret these findings in the broader context of QA hardware and noise processes. 
We first discuss the observed sensitivities with respect to annealing schedules and chain strength, linking them to our theoretical framework. 
We then position our model relative to other physical noise channels that are not explicitly represented by ICE, highlighting directions for more comprehensive noise modeling.
\subsection{Schedule and chain-strength sensitivity}

The combined results of Section~\ref{res:RQ3} and Section~\ref{res:RQ4} provide a unified perspective on how annealing time and chain strength influence chain stability. 
With respect to the annealing schedule, we observe that longer annealing times systematically reduce the absolute level of CBF while leaving the scaling with embedding size unchanged. 
This suggests that extended schedules mitigate break events by enhancing adiabaticity, but the fundamental growth of CBF with chain length is governed primarily by accumulated control errors.

In contrast, varying the chain strength $k$ directly modulates the stabilizing margin. 
Larger $k$ values consistently suppress CBF and simultaneously improve concordance between observed and predicted values, as reflected in the decreasing SSE in Table~\ref{tab:fit_params_k}. 
Correspondingly, the fitted margin parameter $\kappa$ grows nearly linearly with $k$, which we interpret as the empirical realization of the effective chain strength, $k_{\mathrm{eff}} = \eta k$.
Within this framework, the accumulated error $\Delta(\ell_i)$ represents the total control error along a chain corresponding to logical variable $i$, of length $\ell_i$, arising from both local field and coupler perturbations, and is modeled as~\eqref{eq:ref1}.

As established in proof of Theorem~\ref{th:gaussian-tail}, the CBP can be expressed in closed form as a complementary error function of the ratio $k_{\mathrm{eff}}/\sqrt{\mathrm{Var}[\Delta(\ell_i)]}$. 
In the asymptotic regime, $\operatorname{erfc}(x)$ decays as $e^{-x^2}$ up to a polynomial prefactor, so that chain breaks are exponentially suppressed with increasing $k_{\mathrm{eff}}$.
This asymptotic form provides a clear explanation for why increasing $k_{\mathrm{eff}}$ exponentially suppresses chain breaks in practice. 
Moreover, at larger $k$ values, break events are governed almost entirely by this Gaussian-tail mechanism, reducing the influence of unmodeled effects, such as ghost couplings or calibration drift~\cite{yarkoni2022quantum}. 
This accounts for the progressively smaller fitting errors observed in the large-$k$ regime.

Nevertheless, $k$ cannot be increased without bound. 
As shown in Definition~\ref{def:emHamil}, the embedded Hamiltonian assigns the original logical couplers $J_{ij}$ to physical edges between chains, while also introducing chain strength (intra-chain penalty) $k$ to enforce chain consistency. 
If $k$ is chosen excessively large, the intra-chain penalty dominates relative to the logical couplers $J_{ij}$, which encode the intended problem interactions between logical variables. 
As a result, the embedded Hamiltonian underrepresents the original problem Hamiltonian, reducing logical fidelity. 
This trade-off manifests clearly in our experiments: too small $k$ produces unstable chains with high CBF, while too large $k$ stabilizes chains at the expense of faithfully representing logical interactions.
Our empirical results show that the critical chain strength follows a sublinear power-law scaling with chain length, $k^\ast(\ell_i) \propto \ell_i^\alpha$, with fitted exponents $\alpha\simeq 0.82$-$0.85$. 
This is substantially steeper than the theoretical $\alpha=0.5$ predicted by the $\sqrt{\ell_i}$ law under independent Gaussian error assumptions.

This deviation may be explained by the presence of correlated or global noise sources in addition to independent control errors.
To capture this effect, we extend the variance model as
\begin{equation}
\mathrm{Var}[\Delta(\ell_i)] = \ell_i\sigma_h^2 + (\ell_i-1)\sigma_c^2 + \rho\,\ell_i^\gamma,
\end{equation}
where the final term represents a correlated contribution with strength $\rho$ and scaling exponent $\gamma$. 
If this contribution dominates at large $\ell_i$, then
\begin{equation}
k^\ast(\ell_i) \sim \ell_i^{\gamma/2},
\end{equation}
implying an effective scaling exponent $\alpha = \gamma/2$. 
Our fitted $\alpha \approx 0.82$ corresponds to $\gamma \approx 1.64$, consistent with partially correlated error processes that accelerate the growth of required chain strength. 
Such mechanisms may include calibration drift, flux offsets, or Hamiltonian programming noise~\cite{zaborniak2021benchmarking, pelofske2023noise}. 

In this perspective, the ideal $\sqrt{\ell_i}$ law serves as a lower bound achievable under independent Gaussian errors, whereas real hardware exhibits a steeper yet still sublinear scaling due to correlated noise. 
This establishes a practical design rule: chain strength must grow faster than $\sqrt{\ell_i}$ but remains well below linear growth. 
Consequently, there exists a sweet spot in $k$ where chain reliability is maintained without overwhelming the logical couplers, and embedding-aware calibration strategies become essential for achieving robust QA at scale.

\subsection{Relation to other physical noise channels}
Our embedding-aware analysis focuses on ICE, 
which model analog misspecification of fields and couplers through Gaussian perturbations 
$\delta h_i$ and $\delta J_{ij}$. 
In practice, ICE serves as a convenient summary equation that absorbs multiple hardware imperfections into effective noise terms~\cite{yarkoni2022quantum}. 
However, real quantum annealers are also influenced by additional physical channels that are not explicitly captured in our model. 
A complete embedding-aware noise framework will need to incorporate these effects:

\begin{itemize}
    \item Thermal noise:  
    due to finite operating temperature, qubits can undergo thermally activated hops during annealing. 
    These excitations compete with quantum tunneling and can push the system into higher-energy states, 
    particularly near freeze-out when dynamics effectively stop~\cite{raymond2016global, dickson2013thermally}. 
    The influence of thermal noise depends strongly on temperature, minimum gap, and annealing schedule.

    \item Decoherence and environment coupling:
    QA suffers from nonadiabatic transitions and decoherence~\cite{aaberg2005quantum, childs2001robustness}. This decoherence induces unwanted transitions from the ground state to excited states. For the decoherence noise model, the quantum adiabatic Markovian master equation~\cite{albash2012quantum} can be introduced with the total Hamiltonian for the system and its environment: 
    \begin{equation}
    H_{\text{total}}(t)=H_{\text{sys}}(t)+H_{\text{env}}+H_{\text{int}},        
    \end{equation}
    where $H_{\text{sys}}(t)$ denotes the annealing system defined by~\eqref{eq:qaorigin}, $H_{\text{env}}$ represents the environment (often called a {bath} in open quantum systems, e.g., a thermal reservoir or external modes), and $H_{\text{int}}$ is the interaction Hamiltonian coupling the system to this environment. While QA is relatively robust to pure dephasing in the computational basis, decoherence shortens coherence times and alters the effective annealing trajectory, with stronger effects at longer schedules~\cite{boixo2013experimental}.
\end{itemize}

In summary, while ICE provides a compact representation of analog control errors via $\delta h$ and $\delta J$, 
a fully embedding-aware noise model must connect these terms to underlying physical noise processes. 
This would enable predictive models that remain faithful to device physics and lay the groundwork for hardware-aware QA simulators.

\section{Conclusion}\label{sec:conclusion}
This paper has proposed an embedding-aware noise model that extends the ICE framework to explicitly capture how embedding overhead amplifies hardware noise in QA. 
By modeling control errors as Gaussian perturbations on fields and couplers, we derived closed-form scaling laws for variance growth, CBP, and CBF. 
These analytical predictions were experimentally validated on the D-Wave Advantage2 system (Zephyr topology) using clique embeddings of random QUBOs, where the observed consistency between predicted and measured CBF confirmed that embedding length is a primary driver of reliability degradation in current hardware.

Our analyses have further demonstrated that chain strength plays a stabilizing role: the fitted margin parameter $\kappa$ scales nearly linearly with $k$, consistent with the proposed failure criterion. 
This supports the theoretical design rule that chain strength should scale as $\sqrt{\ell_i}$ to maintain constant reliability as embeddings grow, while avoiding the over-suppression of logical couplers. 
At the same time, hardware experiments revealed a steeper but still sublinear scaling exponent ($\alpha\simeq 0.82$-$0.85$), which we attribute to correlated or global noise processes. 
This contrast between the idealized theoretical law and the practical scaling law highlights both the validity of the framework and the importance of considering additional hardware imperfections.

Beyond schedule and chain strength sensitivity, our results have established the generality of the embedding-aware Gaussian noise model. 
Nevertheless, the present analysis isolates only embedding-induced control errors and does not explicitly capture other physical noise channels, such as thermal excitations, decoherence, and system-environment couplings. 
Future work should therefore extend embedding-aware noise modeling to incorporate these processes into the perturbation terms $\delta h$ and $\delta J$, and to consider correlated and non-Gaussian error distributions. 
Such extensions would enable the development of hardware-aware QA simulators that combine analytic embedding laws with realistic physical noise, ultimately providing predictive insight into scalability and informing embedding-aware parameter tuning strategies for next-generation quantum annealers.

\section*{Acknowledgements}
This work was supported in part by Quantum Computing based on Quantum Advantage challenge research(RS-2024-00408613) through the National Research Foundation of Korea(NRF) funded by the Korean government (Ministry of Science and ICT(MSIT)); in part by Creation of the quantum information science R\&D ecosystem(based on human resources) through the National Research Foundation of Korea(NRF) funded by the Korean government (Ministry of Science and ICT(MSIT)) (RS-2023-00256050); in part by the National Research Foundation of Korea(NRF) grant funded by the Korea government(MSIT)(RS-2024-00336962); and in part by the Institute of Information and communications Technology Planning and Evaluation (IITP) under the Artificial Intelligence Convergence Innovation Human Resources Development (IITP-2025-RS-2023-00254177) grant funded by the Korea government(MSIT).  

\bibliographystyle{quantum}
\bibliography{bibliography}

\onecolumn
\appendix

\end{document}